\documentclass[preprint]{imsart}
\usepackage{amsthm, amsmath, amsfonts, natbib, graphicx}
\RequirePackage[colorlinks,citecolor=blue,urlcolor=blue]{hyperref}

\arxiv{1311.3190}
\graphicspath{{figures/}}

\startlocaldefs
    \numberwithin{equation}{section}
    \theoremstyle{plain}
    \newtheorem{thm}{Theorem}[section]
    
    \newtheorem{corollary}{Corollary}[section]
    \newtheorem{lemma}{Lemma}[section]

    \theoremstyle{definition}
    \newtheorem{definition}{Definition}[section]

    \theoremstyle{remark}
    \newtheorem{remark}{Remark}[section]

    \newcommand{\e}[1]{\mathbb{E}[#1]}
    \newcommand{\var}[1]{Var(#1)}
    \newcommand{\simiid}{\overset{i.i.d.}{\sim}}
    \newcommand{\pr}[1]{\Pr \left[#1\right]}
    \newcommand{\ntoinf}{n \rightarrow \infty}
    \newcommand{\X}{{\bf X}}
    \newcommand{\comment}[1]{}
   
    \DeclareMathOperator*{\argmin}{\arg\!\min}
    
\endlocaldefs

\begin{document}

\begin{frontmatter}
    \title{On the exact Berk-Jones statistics and their \(p\)-value calculation}
    \runtitle{Exact BJ statistics}

    \begin{aug}
        \author{
            \fnms{Amit}
            \snm{Moscovich}
            \thanksref{a,e1}
            \ead[label=e1,mark]{amit.moscovich@weizmann.ac.il}
        }
        \author{
            \fnms{Boaz}
            \snm{Nadler}
            \thanksref{a,e2}
            \ead[label=e2,mark]{boaz.nadler@weizmann.ac.il}
        }
        \and
        \\
        \author{
            \fnms{Clifford}
            \snm{Spiegelman}
            \thanksref{b,e3}
            \ead[label=e3,mark]{cliff@stat.tamu.edu}
        }
        \address[a]{
            Department of Computer Science and Applied Mathematics,
            Weizmann Institute of Science, Rehovot, Israel.
            \printead{e1,e2}
        }
        \address[b]{
            Department of Statistics, Texas A\&M University, College Station TX, USA.
            \printead{e3}
        }

        \runauthor{Moscovich, Nadler, Spiegelman}

        \affiliation{Some University and Another University}
\end{aug}

    \begin{abstract}
        Continuous goodness-of-fit testing is a classical problem in statistics. 
        Despite having low power for detecting deviations at the tail of a distribution, the most popular test is based on the Kolmogorov-Smirnov statistic. 
        While similar variance-weighted statistics, such as Anderson-Darling and the Higher Criticism statistic give more weight to tail deviations, as shown in various works, they still mishandle the extreme tails.
      
        As a viable alternative, in this paper we study some of the statistical properties of the  exact $M_n$ statistics of Berk and Jones.
        In particular we show that they are consistent and asymptotically optimal for detecting a wide range of rare-weak mixture models.
        Additionally, we present a new computationally efficient method to calculate $p$-values
        for any supremum-based one-sided statistic, including the one-sided $M_n^+,M_n^-$ and $R_n^+,R_n^-$ statistics
        of Berk and Jones and the Higher Criticism statistic.
        Finally, we show that $M_n$ compares favorably to related statistics in several finite-sample simulations.
    \end{abstract}

\end{frontmatter}

\section{Introduction}

Let $x_1,x_2,\ldots,x_n$ be a sample of $n$ i.i.d. observations of a real-valued one-dimensional random variable \(X\).
The classical continuous goodness-of-fit (GOF) problem is to assess the validity of a null hypothesis that \(X\) follows a known
(and fully specified) continuous distribution function $F$, against an unknown and arbitrary alternative \(G\),  
\begin{equation}
    \mathcal H_0: X\sim F\quad vs. \quad \mathcal H_1:X\sim G 
    \ \mbox{ with }\ G\neq F .
\end{equation}

Goodness-of-fit is one of the most fundamental hypothesis testing problems  \citep{LehmannRomano2005}.
Most GOF tests for continuous distributions
can be broadly categorized into two groups. The first comprises of tests based on some distance metric between
the null distribution \(F\) and the empirical distribution function
$\hat{F}_n(x) = \frac{1}{n} \sum_i \mathrm{\bf 1}(x_i \le x)$.
These include, among others,
the tests  of Kolmogorov-Smirnov (KS),  Cram\'er-von Mises, Anderson-Darling (AD),
Berk-Jones, as well as the Higher Criticism (HC) and Phi-divergence tests \citep{AndersonDarling1954, BerkJones1979, JagerWellner2007}.
The second group considers the first few moments of the random variable \(X\)
with respect to an orthonormal basis of \(L_{2}(\mathbb{R})\).
Notable representatives are Neyman's smooth test \citep{Neyman37}, and its more recent data-driven versions, where the number of moments is determined in an adaptive manner, see \citet{Ledwina1994} and \citet{Rayner_book}.  

Despite the abundance of GOF tests,
KS is nonetheless the most commonly used in practice.
It has several desirable properties,
including asymptotic consistency against any fixed alternative,
good power against a shift in the median of the distribution \citep{Janssen2000},
and the availability of simple procedures to compute its $p$-value. 
However, it suffers from a well known limitation -- it
has little power for detecting deviations at the tails 
of the distribution,
which is important in a variety of practical situations. One scenario is the detection of rare contaminations, whereby only a few of the \(n\) observations are contaminated and arise from a different distribution.
A specific example is  the {\em rare-weak model}  \citep{Ingster1997,DonohoJin2004} and its generalization to sparse mixture models  \citep{CaiWu2014}. 
 Another example involves high dimensional variable selection or multiple hypothesis testing problems under sparsity assumptions \citep{Walther}. 

Given the popularity of the KS test, a natural question is how can it be modified
to have tail sensitivity, and what are the properties of the resulting test.
In this paper we make several contributions regarding these questions.
We start in Section \ref{sec:KS_AD_HC} by viewing the KS and the variance-weighted AD and HC statistics under a common framework, as different ways to measure
the deviations of order statistics  from their expectations.
As described in Section \ref{sec:M_n},
this leads us to study a different GOF statistic, based on the following principle:  
Rather than looking for the largest (possibly weighted) deviation, it looks for the deviation
which is  {\em most statistically significant}. 
Independently of our work, equivalent GOF tests were recently suggested by several different authors, including \citet{MaryFerrari2014, Finner2014Bernoulli, KaplanGoldman2014}.
This statistic is also closely related to the work of \cite{Brown2013}
who instead of a GOF test, derived a method to construct confidence bands for a Normal Q-Q plot.
It turns out, however, that all of these proposals are in fact \textit{equivalent} to GOF testing based on the $M_n$ statistic defined in \citet{BerkJones1979}.
The \(M_n\) statistic was  derived based on an earlier work by the authors on relatively optimal combinations of test statistics \citep{BerkJones1978}. The  $R_n$ statistic
(often called \emph{the} Berk-Jones statistic) was then proposed as an approximation to \(M_n\), which is simpler to compute. 
However, with today's computers, this approximation is no longer necessary and the $M_n$ statistic can be computed directly.

On the theoretical front, in Section \ref{sec:theoretical_properties} we analyze some statistical properties of the $M_n^+, M_n^-$ and $M_n$.
First, we derive the asymptotic distribution of these statistics under the null hypothesis, our proof is based on classical results from
the theory of standardized empirical processes \citep{Eicker1979, Jaeschke1979}.
Independent of our work, a different derivation was recently given by \cite{Finner2015Communications},
based on an analysis of the HC statistic \citep{Finner2014Bernoulli}.
Next, we use the asymptotic distribution to prove asymptotic consistency of $M_n$ against any fixed alternative $G \neq F$,
as well as against series of converging alternatives $G_n \rightarrow F$ provided that the convergence
in the supremum norm $\|G_n - F \|_{\infty}$ is sufficiently slow.
Finally, following the work of \citet{CaiWu2014} we show that $M_n$  is adaptively optimal for detecting a broad family of sparse mixtures. 
  
In a second contribution, we devise in Section \ref{sec:pvalue_calc} an $O(n^2)$ algorithm to compute $p$-values for  {\em any} supremum-based one-sided test. Particular examples include HC as well as the one-sided $M_n^\pm$ and $R_n^\pm$ statistics of Berk and Jones.

Finally, in Section \ref{sec:simulations} we compare
the power of $M_n$ to other tests under the following  settings: i) a change in the mean or variance
of a standard Gaussian distribution; and ii)
rare-weak sparse Gaussian mixtures; 
These results showcase scenarios where $M_n$ has improved power compared to common tests. For other examples involving real data and concrete applications, see \cite{Brown2013,Siegmund_14,KaplanGoldman2014}.

\section{The Kolmogorov-Smirnov, Anderson-Darling and Higher Criticism Statistics}
\label{sec:KS_AD_HC}

Let us first introduce some notation.
For a given sample \(x_{1},\ldots,x_n\), we denote by \(x_{(i)}\) the $i$-th sorted observation (i.e. $x_{(1)} \leq x_{(2)} \leq \ldots \leq x_{(n)}$), by \(u_{i}=F(x_i)\), and by \(u_{(i)}=F(x_{(i)})\) where $F$ denotes the null distribution. Finally, we denote the empirical distribution by \(\hat F_n(x)=\frac1n\sum_i \mathrm{\bf 1}(x_i \le x)\).

The standard definition of the KS test statistic is based on a (two-sided) $L_{\infty}$ distance over a continuous variable $x\in\mathbb{R}$, 
\begin{align}
    K_n := \sqrt{n} \sup_{x \in \mathbb{R}} \left| \hat{F}_n(x) - F(x) \right| . \label{eq:ks_continuous}
\end{align}
Although Eq. \eqref{eq:ks_continuous} involves a supremum over \(x\in\mathbb{R}\),  in what follows we instead use an equivalent discrete formulation, whereby the two-sided KS statistic is the maximum of a pair of discrete one-sided statistics,  $K_n := \max (K_n^{-}, K_n^{+})$, where 
\begin{align}
    K_n^{-} &:= \sqrt{n} \max_{i} \left( u_{(i)} - \frac{i-1}{n} \right),
        \quad
    K_n^{+} := \sqrt{n}  \max_{i} \left( \frac{i}{n} - u_{(i)} \right) .
        \label{eq:K_n_pm}
\end{align}
By the definition of $\hat{F}_n$, under the null hypothesis that all $x_i \sim F$ we have
\[
    n\hat F_n(x)\sim Binomial(n,F(x)) \qquad  \forall x \in \mathbb{R}.
\]
Hence, $\mathbb{E}[\hat F_n(x)]=F(x)$ and
$Var[\hat F_n(x)]=\frac{1}{n} F(x)(1-F(x))$. The latter varies significantly throughout the range of $x$,
attaining a maximum at the median of the distribution and smaller values near the tails.  

\citet{AndersonDarling1952} were among the first to suggest different weights to deviations at different locations. Based on a weight function $\psi:[0,1] \rightarrow \mathbb{R}$. they proposed a weighted $L_2$ statistic 
\begin{align} \label{eq:AD_definition}
    \mbox{AD}_{n,\psi} &= \int_{-\infty}^{+\infty} n\left(\hat{F}_n(x) - F(x)\right)^2 \psi(F(x)) f(x) dx \,,
\end{align}
and a lesser-known weighted $L_\infty$ statistic, defined as
\begin{align}
    \label{eq:ADsup_definition}
    \mbox{AD}^{sup}_{n,\psi} &= \sup_{x \in \mathbb{R}} \sqrt{n} |\hat{F}_n(x) - F(x)| \sqrt{\psi(F(x))}  \,.   
\end{align}
Specifically, \cite{AndersonDarling1952} suggested to use the weight function $\psi(x) = \tfrac1{x(1-x)}$ which standardizes the variance of $\hat{F}_n(x)$.

Closely related to Eq. (\ref{eq:ADsup_definition}) is the Higher Criticism statistic, whose two variants below can be viewed as one-sided GOF test statistics,
\begin{align}
    \mbox{HC}_n^{2004} &:= \sqrt{n} \max_{1 \le i \le \alpha_0 \cdot n}  \frac{\frac{i}{n} - u_{(i)}}{\sqrt{u_{(i)} (1 - u_{(i)})}} \quad &\textrm{\citep{DonohoJin2004}}, \label{eq:HC_04} \\
    \mbox{HC}_n^{2008}&:=  \sqrt{n} \max_{1 \le i \le \alpha_0 \cdot n}  \frac{\frac{i}{n} - u_{(i)}}{\sqrt{\frac{i}{n} (1 - \frac{i}{n})}} \quad &\textrm{\citep{DonohoJin2008}}. \label{eq:HC_08}
\end{align}
Indeed, the $\mbox{HC}_n^{2004}$ test with $\alpha_0 = 1$ is equivalent to a one-sided variant of the $\mbox{AD}^{sup}_{n,\psi}$
test with $\psi(x) = 1/x(1-x)$. 

\subsection{Order Statistics of Uniform Random Variables} \label{sec:order_statistics}

By the {\em probability integral transform},
if $X\sim F$ with $F$ a continuous cdf, then $Y=F(X)$ follows a uniform distribution $Y \sim U[0,1]$.
Hence, under the null, the transformed values \(u_{i}=F(x_i)\) are an i.i.d. sample from the \(U[0,1]\) distribution
and the sorted values $u_{(i)} = F(x_{(i)})$ are their {\em order statistics}. In particular, the distribution of the $i$-th order statistic, $U_{(i)}$, is given by 
\begin{align}
    U_{(i)} \sim \textrm{\emph{Beta}}(i,n-i+1) ,  \label{eq:U_i_Beta}
\end{align}
\noindent with the following mean and variance
\begin{equation} \label{eq:E_V_Beta}
    \e{U_{(i)}} = \frac{i}{n+1} \qquad
    \var{U_{(i)}} = \frac{i (n-i+1)}{(n+1)^2 (n+2)}.        
\end{equation}
We now relate the KS and HC tests to $U[0,1]$ order statistics. 
Up to a small $O(1/\sqrt{n})$ correction, 
the one sided KS statistic of Eq. (\ref{eq:K_n_pm}) is the maximal deviation of the \(n\)
different uniform order statistics from their expectations,  
\begin{equation}
    \label{eq:KS_uniform}
    K_{n}^+  = \max_{i} \sqrt{n}\left(\e{U_{(i)}}-u_{(i)}   \right)+ O\left(\tfrac1{\sqrt{n}}\right) \,.
\end{equation}
The variance of each $U_{(i)}$ is different, with a maximum at \(i=n/2\).
Hence the largest deviation tends to occur 
near the center. Importantly, such deviations can mask small, but statistically significant, deviations at the tails, leading
to poor tail sensitivity
\citep{MasonSchuenemeyer1983, Calitz1987}.

In contrast, up to a small correction term, the HC$^{2008}$ statistic normalizes the difference $\mathbb{E}[U_{(i)}]-u_{(i)}$ by its standard deviation,
\begin{equation}
    \label{eq:_uniform}
    \mbox{HC}_{n}^{2008}=\max_i \sqrt{n} \frac{i/n-u_{(i)}}{\sqrt{i/n (1-i/n)}} = \max_i \frac{\e{U_{(i)}} - u_{(i)} }{stdev[U_{(i)}]} \left(1+O(1/n)\right)\,,
\end{equation}
and the HC$^{2004}$ / AD$^{sup}$ statistics perform a similar normalization.
Such normalizations are common when comparing Gaussian variables with different variances. 
Indeed, at indices $1 \ll i \ll n$, the distribution of $U_{(i)}$ is close to Gaussian.
However, this is not the case when $i$ is fixed and $n \to \infty$ \citep{KeilsonSumita1983}.
In particular, for any \(n\geq 2\) the distribution of $U_{(1)}$ is monotone and heavily skewed towards zero. In section \ref{sec:sim_gaussian} we  demonstrate and explain analytically why the normalization (\ref{eq:_uniform})
can adversely affect the detection power of HC.
%

\section{The exact Berk-Jones statistics}

\label{sec:M_n}
The discussion above demonstrates that both the KS and HC\ statistics
do not uniformly 
calibrate the deviations $u_{(i)}$ over the entire range $i \in \{1, \ldots, n \}$.
In this paper we study the $M_n, M_n^+$ and $M_n^-$ statistics, whose key underlying principle
can be described as looking for the deviation \(\e{U_{(i)}}-u_{(i)}\) which is 
{\em most statistically significant}. In details, for each transformed order statistic $u_{(i)}$, we first compute
a one-sided $p$-value, according to its null distribution Beta$(i,n-i+1)$.
This \(p\)-value is given by
\begin{equation}
    \label{def:p_i}
    p_{(i)} := \pr{\text{Beta}(i,n-i+1) < u_{(i)}}\,,
\end{equation}
Then, in analogy to KS, 
we define the one-sided $M_n^-, M_n^+$ and two-sided $M_n$ statistics by
\begin{equation}
M_n^+ := \min_{1 \le i \le n} p_{(i)}, \ \  M_n^- := \min_{1 \le i \le n} \left( 1-p_{(i)} \right)          
\ \mbox{and}\  
M_n := \min\{M_n^+, M_n^-\}. \label{eq:def_Mn}
\end{equation}
In contrast to the KS statistic, whose range is \([0,\infty) \) and for which large values lead to a rejection of the null,
the \(M_{n}\) statistic is always in $[0,1]$,
with \emph{small values} indicating 
a \emph{bad fit} to the null hypothesis.
Note that $p_{(i)} = I_{u_{(i)}}(i,n-i+1)$, where $I_x(\alpha, \beta)$
is the regularized incomplete Beta function. This function is commonly available in standard mathematical packages, hence the numerical evaluation of the statistics $M_n$ and $M_n^{\pm}$ is straightforward.

Independently of our work, test procedures of the form $M_n < c$ have been recently suggested in several different papers, including
\citet{MaryFerrari2014, KaplanGoldman2014, Finner2014Bernoulli}. However, a close examination reveals that the definitions
in Eq. \eqref{eq:def_Mn} 
are in fact \textit{equivalent} to those proposed by \citet{BerkJones1979}.
In contrast to our motivation, their derivation of $M_{n}$  followed a different path,
building upon their earlier work on relatively optimal combinations of test statistics \citep{BerkJones1978}.

\citet{BerkJones1979} also defined the $R_n, R_n^+$ and $R_n^-$ statistics,
as approximations to the $M_n, M_n^+$ and $M_n^-$ statistics.
At the time, this was necessary because computers and software to calculate the tails of
a Beta distribution were not as widespread as today.
As a result, the approximate statistics became known as \emph{the} Berk-Jones statistics,
whereas the exact $M_n$ statistics seem to have received far less attention.
With today's widespread availability of computers, direct calculation of the exact statistics poses no difficulty, and their approximation
is no longer necessary.

In the following sections we derive the asymptotic null distribution of the $M_n, M_n^+$ and $M_n^-$ statistics,
present an $O(n^2)$ numerical procedure to compute exact $p$-values for $M_n^+$\ and $M_n^-$,
and empirically compare their detection power to other GOF tests in several simulations. 
\subsection{Confidence Bands}
\label{subsec:confidence_bands}
Often, one is interested not only in the magnitude of the most statistically significant deviation
from the null hypothesis, as can be measured by  $M_n$ or other statistics,
but also in gaining insight into the nature of the deviations throughout the entire range of the sample set.
One common practice is to draw a Q-Q scatter plot of the points $\{(F^{-1}(\tfrac{i}{n+1}), x_{(i)})\}_{i=1}^{n}$.
From Eq. \eqref{eq:E_V_Beta} it follows that under the null $F(x_{(i)}) = u_{(i)} \approx \tfrac{i}{n+1} $, and hence the Q-Q plot should be concentrated around the $x=y$ diagonal.

Similar to \cite{Owen1995}, who constructed $\alpha$-level confidence bands around the diagonal based on the $R_n$ statistic,
one can instead use the $M_n$ statistic.
Let $c_\alpha \in [0, 1]$ be the $M_n$ threshold that corresponds to an $\alpha$-level test. i.e.
\[
    \Pr[M_n < c_\alpha | \mathcal{H}_0] = \alpha.
\]
By definition \eqref{eq:def_Mn}, $M_n > c_\alpha$ if and only if the transformed order statistics all satisfy $ b_i < u_{(i)} < B_i$ where $b_i$ and $B_i$
are the $c_\alpha$ and $1-c_\alpha$ quantiles of the Beta$(i,n-i+1)$ distribution, respectively.
Upon making the inverse transformation $x_{(i)} = F^{-1}(u_{(i)})$, this yields confidence bands for the entire Q-Q plot.
In the Gaussian case, these confidence bands are precisely those of \cite{Brown2013}.
For a related construction of confidence bands and further discussion, see \cite{Dumbgen_Wellner_14}. 

\section{Theoretical Properties of the exact Berk-Jones statistics}
\label{sec:theoretical_properties}
Theorem \ref{thm:Mnplus_asymptotic_distribution} below provides the exact asymptotic null distribution of the $M_n^+,M_n^-$ and $M_n$ statistics.
\begin{thm} \label{thm:Mnplus_asymptotic_distribution}
    Under the null hypothesis, for any fixed $x > 0$
    \begin{align*}
        &\pr{M_n^{\pm} < \frac{x}{2\log n \log \log n} \bigg| \mathcal{H}_0} \xrightarrow{n \to \infty} 1 - e^{-x} \\
        &\pr{M_n \hspace{0.85mm} < \frac{x}{2\log n \log \log n} \bigg| \mathcal{H}_0}       \xrightarrow{n \to \infty} 1 - e^{-2x}
    \end{align*}
\end{thm}
\noindent We note that this result was recently proved by \citet{Finner2015Communications}.
Their proof is based on a detailed analysis of the local levels of the HC statistic \citep{Finner2014Bernoulli}.
We present a different proof, which approximates the distribution of each $U_{(i)}$ by a Gaussian variable and adapts
known results from the theory of standardized empirical processes.
This theorem enables one to construct asymptotic $\alpha$-level tests and prove the consistency of tests based on the $M_n$ statistic.
In fact, we show that $M_n$ is consistent even against a series of converging 
alternatives $G_n \xrightarrow{\ntoinf} F$, provided that this convergence is sufficiently slow.
Similar properties hold for KS and are considered desirable for any GOF statistic \citep[Chapter 14]{LehmannRomano2005}. 

In Section \ref{sec:sparse_mixtures} below, we show that the $M_n$ statistic is asymptotically optimal for 
detecting deviations from a Gaussian distribution for a wide class of rare-weak contamination models. 

We note that the asymptotic distribution of the approximate Berk-Jones statistic $R_n$
was already given in \cite{BerkJones1979}, where a sketch of a proof appears. \cite{WellnerKoltchinskii2003} provide a rigorous proof.
This result was generalized by \cite{JagerWellner2007} for a wider class of GOF statistics based on phi-divergences.

\subsection{Asymptotic Consistency of $M_n$}
\label{subsec:M_n}

Next, we study the asymptotics of $M_n$ under various alternatives.
First, we consider the case of a \emph{fixed} alternative.
\begin{thm}
    \label{thm:Mn_H_1}
    Let $X_1, \ldots, X_n \simiid G\neq F$. Then, for any \(\epsilon > 0 \)
    \begin{align}
        \label{eq:Mn_H_1}
        \pr{
            M_n(F(X_1), \ldots, F(X_n))  <  \frac{1+\epsilon}{4\| G-F \|_\infty^2} \cdot \frac{1}{n}
            \Big| \mathcal{H}_1
        }
        \xrightarrow{n \rightarrow \infty}
        1\,.
    \end{align}
\end{thm}
\noindent Combining Theorems \ref{thm:Mnplus_asymptotic_distribution} and \ref{thm:Mn_H_1}, we obtain the following key result.
\begin{corollary}
    $M_n$ is consistent against any fixed alternative.   
\end{corollary}
In other words, as $\ntoinf$ the $M_n$ statistic  
perfectly distinguishes between the null hypothesis $F$
and any fixed alternative $G \neq F$.
In fact, as the following corollary shows, $M_n$ even distinguishes between $F$ and a series of 
converging alternatives $\{G_n\}_{n=1}^{\infty}$ such that $G_n \rightarrow F$,
provided that this convergence is sufficiently slow.
\begin{corollary}
    For any fixed $\epsilon > 0$,
a test based on the    $M_n$  statistic is consistent over all alternatives
    $\{G_n\}_{n=1}^{\infty}$ satisfying
    \begin{equation}
        \sqrt{n}\|G_n - F\|_{\infty}\sqrt{\log n\log \log n} \longrightarrow \infty \,.
        \label{eq:G_nF}
    \end{equation}
\end{corollary}

\noindent We note that Berk\&Jones have investigated the limiting behavior of $R_n$ and $M_n$
and showed that under specific conditions on the alternative distribution \citep[Theorem 4.1]{BerkJones1979}
both $R_n^+$ and $-\tfrac1n \log M_n^+$ converge to a constant which depends on the alternative distribution.
These results were greatly extended by \cite{JagerWellner2007} for a family of GOF statistics based on phi-divergences.
In contrast, our Theorem \ref{thm:Mn_H_1} merely gives a stochastic upper bound on $M_n$, but one that does not require
the alternative distribution to satisfy any particular properties.

\subsection{Sparse Mixture Detection}\label{sec:sparse_mixtures}
Motivated by the works of \cite{DonohoJin2004} and \cite{CaiWu2014}, we now study the properties of $M_n$ under the following class of sparse mixture models. Suppose that under the null hypothesis $X_i \simiid F$, whereas under the alternative a \emph{small fraction} $\epsilon_n$ of the variables are contaminated and have a different distribution $G_n$. The corresponding hypothesis testing problem is
\begin{align} \label{H0H1_general_sparse_mixture}
    &\mathcal{H}_0: X_i \simiid F  \quad vs. \quad \mathcal{H}_1: X_i \simiid (1-\epsilon_n) F + \epsilon_n G_n \,.
\end{align}

Such models have been analyzed, among others,  by \cite{Ingster1997}, \cite{DonohoJin2004} and \cite{CaiWu2014}.
Let us briefly review some results regarding these models, first for the \textit{Gaussian mixture model}, where \(F=\mathcal N(0,1)\) and $G_n=N(\mu_n,1)$,
\begin{align} \label{eq:H0H1_sparse_normal_mixture}
    &\mathcal{H}_0: X_i \simiid \mathcal{N}(0,1) \ \ vs. \ \ \mathcal{H}_1: X_i \simiid (1-\epsilon_n) \mathcal{N}(0,1) + \epsilon_n \mathcal{N}(\mu_n, 1) \,.
\end{align}
Recall that for \(n\gg 1\), the maximum of $n$ i.i.d. standard Gaussian variables
is sharply concentrated around $\sqrt{2 \log n}$.
Thus, for any fixed $\epsilon_n=\epsilon$, as $n \to \infty$, contamination strengths $\mu_n > \sqrt{2 \log n}(1+\delta)$ are perfectly detectable by the maximum statistic $\max x_i$.
Similarly, for any fixed $\mu_n=\mu$, sparsity levels $\epsilon_n \gg n^{-1/2}$ visibly shift the overall mean of the samples, and
hence as \(n\to\infty,\) can be perfectly detected by the sum statistic $\sum x_i$. These cases lead one to consider the scaling $\epsilon_n = n^{-\beta}$, $\mu_{n} = \sqrt{2 r \log n}$ and examine the asymptotic detectability in the $(r,\beta)$ plane \citep{Ingster1997}. 
Since any point $(r,\beta)$ with $r>1$ or $\beta<0.5$ is easily detectable,
the interesting  region is where both $0 < r < 1$ and $0.5 < \beta < 1$.

For the model (\ref{eq:H0H1_sparse_normal_mixture}), if $\epsilon_n$ and $\mu_n$ are known, both $\mathcal H_0$ and $\mathcal H_1$ are simple hypotheses, and the optimal test is the likelihood ratio (LR). 
Its performance was studied by \citet{Ingster1997},
who found a sharp detection boundary in the $(r,\beta)$ plane, given by
\begin{align} \label{eq:rareweak_asymptotic_detection_boundary}
    r_{min} (\beta)= \left\{
        \begin{array}{lll}
            \beta - 0.5           & \qquad 0.5&< \beta \le 0.75, \\
            (1-\sqrt{1-\beta})^2  & \qquad 0.75&\le \beta < 1. \\
        \end{array}
    \right.
\end{align}
Namely, as \(n\to\infty\), the sum of type-I and type-II error rates of the LR test tends to 0 or 1 depending on whether $(r,\beta)$ lies above or below this curve.

While the LR test is optimal,  it may be inapplicable as it requires precise knowledge of the model parameters $\mu$ and $\epsilon$. Importantly,
both the Higher Criticism statistic based on Eq. (\ref{eq:HC_04})
and the approximate Berk-Jones test $R_n$ were proven to achieve the optimal asymptotic detection boundary without such knowledge \citep[Theorems 1.2, 1.6]{DonohoJin2004}. Thus, both statistics are \emph{adaptively optimal} for the sparse Gaussian mixture detection problem in an asymptotic sense. In what follows, we prove that $M_n$ is also adaptively optimal.

Recently, \citet{CaiWu2014} studied more general sparse mixtures of the form  \eqref{H0H1_general_sparse_mixture}
where the null distribution is Gaussian and $\epsilon_n =n^{-\beta}$,
but $G_n$ is not necessarily Gaussian.
The following is a simplified version of 
their Theorem 1, describing the asymptotic detectability under this model.

\begin{thm} \label{thm:sparsity_detection_threshold}
    Let $G_n$ be a continuous distribution with density function $g_n$.     
    If the following limit exists for all $u \in \mathbb{R}$ 
    \begin{equation}
        h(u) := \lim_{\ntoinf} \frac{\log \left(\sqrt{2\pi} g_n(u \sqrt{2 \log n})\right)}{\log n}
        \label{eq:alpha_u}
    \end{equation}
    then the hypothesis testing problem \eqref{H0H1_general_sparse_mixture}
    with $F=\mathcal{N}(0,1)$ and $\epsilon_n = n^{-\beta}$ has an asymptotic detection threshold given by 
    \begin{equation}
        \beta^{\star} = \frac{1}{2} + \max \left( 0, \sup_{u \in \mathbb{R}} \left\{h(u) + \frac{1}{2}\min(1,u^2)  \right\} \right) \,.
        \label{eq:beta_star}
    \end{equation}
    Namely, for any $\beta < \beta^\star$ the error rate of the likelihood ratio test tends to zero as $\ntoinf$.    
\end{thm}
In their paper, \citet{CaiWu2014} proved that HC is adaptively optimal under the conditions of Theorem \ref{thm:sparsity_detection_threshold}. As we now show, $M_n$ has the same adaptive optimality properties, in particular for the Gaussian mixture model of Eq. (\ref{eq:H0H1_sparse_normal_mixture}).
We note that for finite sample sizes $M_n$ may have considerably higher power compared to HC as we show in Section \ref{sec:simulations}.

\begin{thm} \label{thm:Mn_adaptive_optimality}
    Let $G_n$ be a continuous distribution satisfying Eq. (\ref{eq:alpha_u}) and let
    \[
        X_1, \ldots, X_n \simiid (1-n^{-\beta}) \mathcal{N}(0,1)+ n^{-\beta} G_n
    \]
    If $\beta < \beta^{*}$ where $\beta^*$ is given by 
    Eq. (\ref{eq:beta_star}), 
    then there is some $\epsilon>0$ such that     \[
        \pr{M_n < \frac{1}{\left( \log n \right)^{1+\epsilon}} } \xrightarrow{\ntoinf} 1.
    \]
    where $\mathcal{N}(0,1)$ is taken as the null distribution.
\end{thm}
\noindent The proof is in the appendix. Combining this result with Theorem \ref{thm:Mnplus_asymptotic_distribution} gives
\begin{corollary}
    For any $\epsilon>0$ and \(\beta<\beta^*\), the test
    \[
        M_n < \frac{1}{\log n (\log \log n)^{1+\epsilon}}
    \]
    perfectly separates, as $\ntoinf$, the null distribution $\mathcal{N}(0,1)$
    from a sparse-mixture alternative of the form $(1-n^{-\beta}) \mathcal{N}(0,1) + n^{-\beta} G_n$. Namely, 
    inside the asymptotic detectability region, the error rate of the test tends to zero. 
\end{corollary}

\section{Computing p-values}

\label{sec:pvalue_calc}

For the classical one-sided and two-sided KS statistics,
there are many methods to compute the corresponding \(p\)-values, see
\citet{Durbin1973,MarsagliaTsangWang2003,BrownHarvey2008one,BrownHarvey2008two}.
Most of these methods, however, are particular to KS and inapplicable to other GOF statistics.
Notable exceptions include \citep{Noe1972, FriedrichSchellhaas1998, KhmaladzeShinjikashvili2001}
whose recursion formulas can compute the $p$-value of any supremum-based two-sided (or one-sided) statistic using $O(n^3)$ operations
and the recent algorithm of \citet{MoscovichNadler2015} that runs in $O(n^2 \log n)$ steps.
For another recent work with time complexity \(O(n^{3})\), see \cite{Barnett_Lin}.

In this section we  present an $O(n^2)$ algorithm to compute \(p\)-values of
any supremum-based \textit{one-sided} test statistic, including $M_n^+, M_n^-, R_n^+, R_n^-$ and the Higher Criticism.
Furthermore, it may be used to obtain approximations of the \(p\)-value of two-sided statistics.

To describe our approach, note that by Eq. \eqref{def:p_i},
\begin{align}
    \pr{M_n^+ \ge c | \mathcal{H}_0} = \pr{\forall i: p_{(i)} \ge c | \mathcal{H}_0} = \pr{\forall i: L^n_{i}(c) \le u_{(i)} \le 1 |\mathcal{H}_0 }, \label{eq:Mn_pvalue} 
\end{align}
where $L^{n}_{i}(c)$ denotes the inverse of the regularized incomplete Beta function, satisfying
\[
    \pr{\text{Beta}(i, n-i+1) < L_i^n(c)} = c.
\]
Procedures to compute \(L_i^n(c)\) are available in most mathematical packages.
Note that under the null,  the \(n\) unsorted variables are uniformly distributed, $U_i \simiid U[0,1]$, and hence their joint density equals 1 inside the \(n\)-dimensional box \([0,1]^{n}\). Given that there are \(n!\) distinct permutations of \(n\) indices, the joint probability density of the random vector of  sorted values $(U_{(1)}, \ldots, U_{(n)})$ is
\[
    f(U_{(1)}, \ldots, U_{(n)}) = 
    \left\{\begin{array}{ll}
        n! & \mbox{if } 0 \le U_{(1)} \le \ldots \le U_{(n)} \le 1, \\
        0  & \mbox{otherwise}. \\
    \end{array}
    \right.
\]
From this it readily follows that
\begin{align}
    \pr{M_n^+ \ge c | \mathcal{H}_0} &=  n!Vol\{ (U_{(1)}, \ldots, U_{(n)}) \ |\ \forall i: L^{n}_{i}(c) \le U_{(i)} \le U_{(i+1)}\} \nonumber \\
    &=  \displaystyle n! \int_{L^{n}_{n}(c)}^1 dU_{(n)} \int_{L^{n}_{n-1}(c)}^{U_{(n)}} dU_{(n-1)} \ldots \int_{L^{n}_{2}(c)}^{U_{(3)}} dU_{(2)} \int_{L^{n}_{1}(c)}^{U_{(2)}} dU_{(1)} \,. \label{eq:Mn_p_value}
\end{align}
Eq. (\ref{eq:Mn_p_value}) is the key to fast calculation of \(p\)-values for $M_n^+$ or other one-sided tests. The idea is to evaluate this multiple integral, from right to left. The first integral yields a polynomial of degree 1 in $U_{(2)}$, the next integral yields a polynomial of degree 2 in $U_{(3)}$ and so on. While we have not found simple explicit formulas for the resulting polynomials, their numerical integration is straightforward. We store \(d+1\) coefficients for the \(d\)-th degree polynomial, and its numerical integration takes  $O(d)$ operations. Hence, the total time complexity is \(O(n^{2})\). 

Still, there are some numerical difficulties with this approach: A na\"{\i}ve implementation suffers from a fast accumulation of numerical errors and breaks down completely at $n \approx 150$. Nonetheless, as described in the appendix, with a  modified procedure and using extended precision (80-bit) floating point numbers, this accumulation of errors is significantly attenuated, allowing accurate calculation of one-sided $p$-values for up to $n \approx 50,000$ samples.
The actual running time of our freely available C++\ implementation is about one second for $n=4000$ samples using a present-day PC.

The following theorem provides simple upper and lower bounds for the \(p\)-value of the two-sided $M_n$, in terms of its one-sided \(p\)-values,  

\begin{thm} \label{thm:two_sided_p_value}
    For any $c\in[0,1]$, let $q_c := \Pr[M_n^+\leq c\ |\ \mathcal{H}_0]$. Then, 
    \begin{equation}
        2q_c-q_c^2 \leq \Pr[M_n \leq c\ |\ \mathcal{H}_0] \leq 2 q_c.
        \label{eq:Mn2_bound}
    \end{equation}
    Furthermore, as $n\to\infty$, 
    \begin{equation}       
        \pr{M_n \leq c \ | \mathcal{H}_0} \xrightarrow{\ntoinf} 2q_c - q_c^2.  \label{eq:Mn2_limit}
    \end{equation}
\end{thm}
\begin{remark}
    As mentioned above, our algorithm can compute the $p$-value of any supremum-type one sided test statistic.
The only difference lies in the  coefficients $L^n_{i}(c)$ of
    Eq. (\ref{eq:Mn_p_value}), which depend on the specific test statistic. For example, the HC$^{2008}$ test of Eq. \eqref{eq:HC_08} satisfies   
    \begin{align*}
        \pr{\mbox{HC}^{2008} < c | \mathcal{H}_0} 
        &= \pr{\forall i: \tfrac{i}{n} - c\sqrt{\tfrac{i}{n^2} \left(1 - \tfrac{i}{n}\right)} < U_{(i)} \leq U_{(i+1)} \Bigg| \mathcal{H}_0}.
    \end{align*}
    Thus, for this statistic, $L_i^n(c) = \frac{i}{n} - c\sqrt{\tfrac{i}{n^2} \left(1 - \tfrac{i}{n}\right)}$.
\end{remark}

\begin{remark}
    Historically, an equation similar to  (\ref{eq:Mn_p_value}) was derived by \cite{Daniels1945},
    in an entirely different context.
    His formula was used in later works to derive closed form expressions for the asymptotic distribution of the KS test statistic.
    See \cite{Durbin1973} for a survey.
\end{remark}

\begin{remark}
    To the best of our knowledge, the only other $O(n^2)$ algorithm
    for computing \textit{p}-values of \(L_{\infty}\)-type one-sided test statistics is that of \cite{KotelnikovaKhmaladze1983}.
    Their method is based on a different recursive formula, which involves  large binomial coefficients
    and  also requires a careful numerical implementation.    
\end{remark}
\section{Simulation Results\label{sec:simulations}}

\begin{figure}
    \centering
    \makebox{\includegraphics[width = 0.48\textwidth]{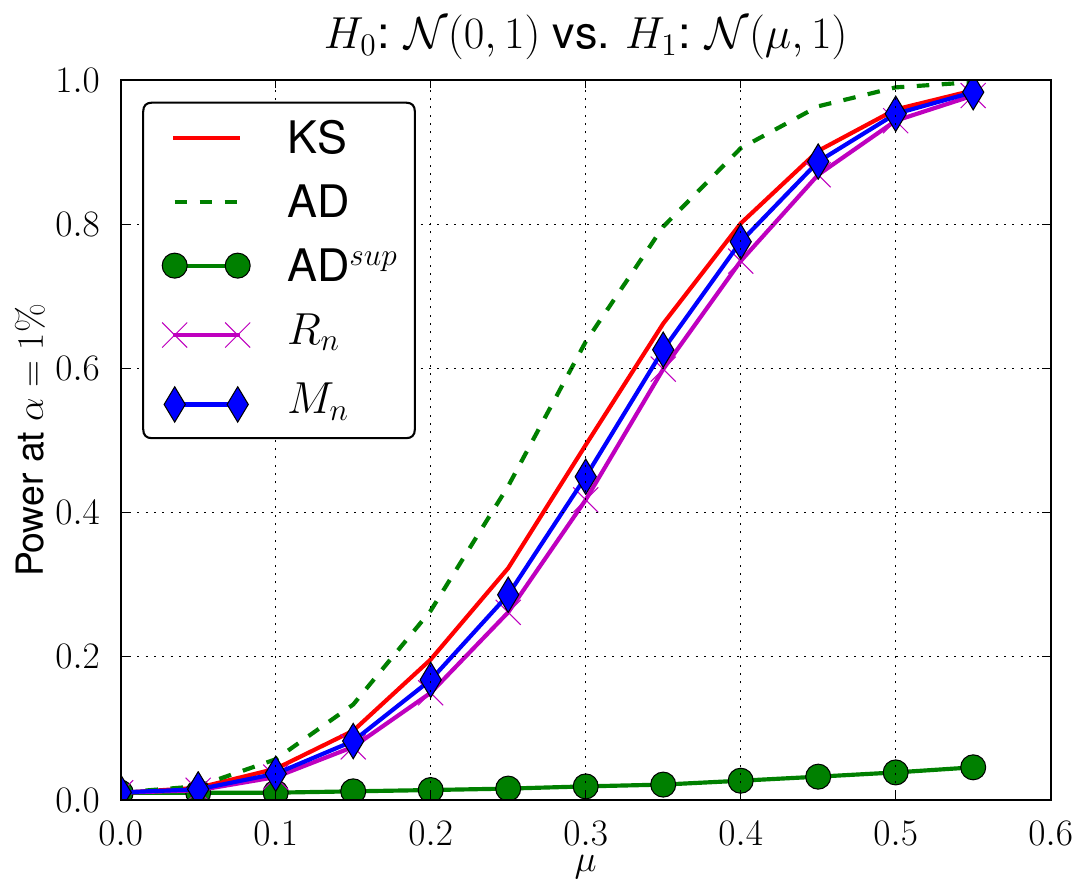}
    \includegraphics[width = 0.48\textwidth]{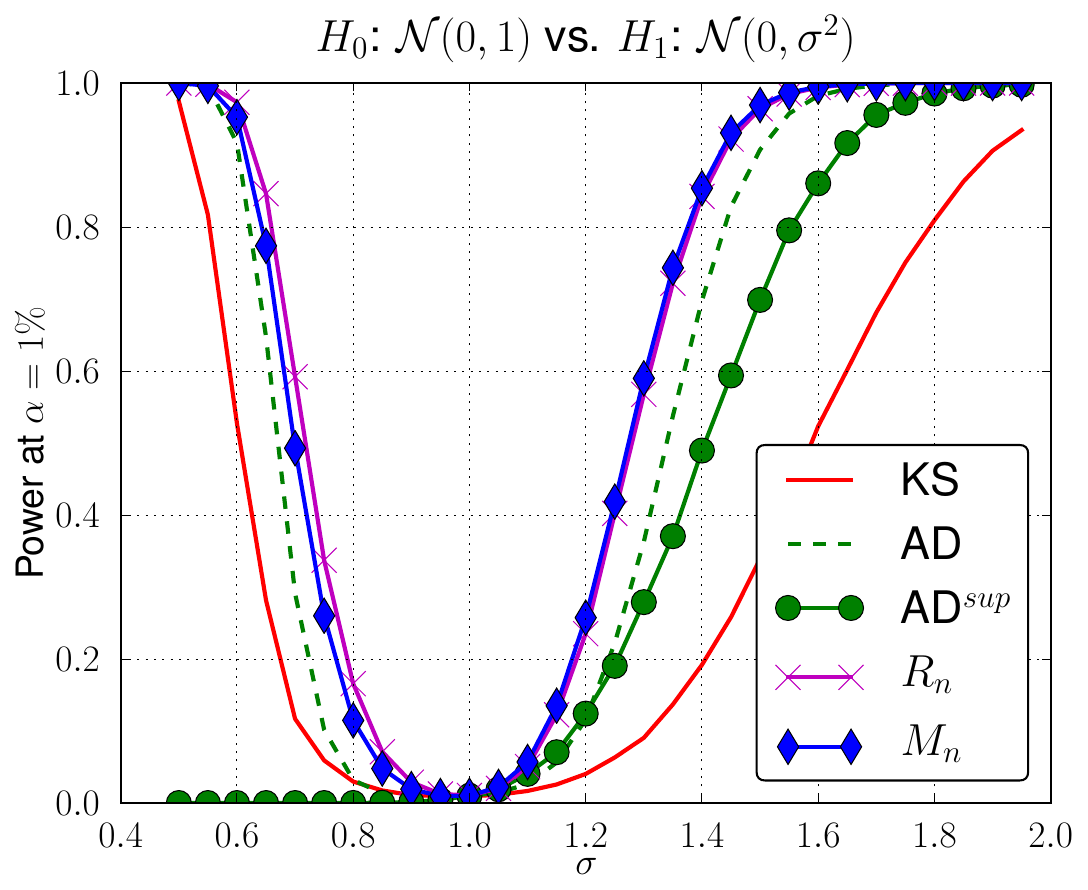}}
    \caption{\label{fig:Gaussian_sim}Power comparisons of two-sided tests for detecting lack-of-fit to a standard Gaussian distribution (at significance level $\alpha =1\%$) with \(n=100\) samples. (left panel) change in the mean of the distribution; (right panel) change in the variance.  }
\end{figure}

\subsection{Deviations from a Standard Gaussian Distribution} \label{sec:sim_gaussian}
We consider a null hypothesis that \(X_{i}\simiid\mathcal N(0,1)\), and two alternatives: a shift in the mean, \(X_{i}\simiid\mathcal N(\mu,1)\), or a change in the variance \(X_{i}\simiid\mathcal N(0,\sigma^2)\). 
The left and right panels of figure \ref{fig:Gaussian_sim} compare the power of $M_n$, KS, AD and AD$^{sup}$ (=two-sided HC) under these two alternatives at a significance level of $\alpha=1\%$.

For detecting a change in the mean, the $M_n$ test is on par with KS, but the AD test outperforms both.
The AD$^{sup}$ test has close to zero power in this benchmark. For detecting a change in the variance, which strongly affects the tails,
$M_n$ has a higher detection power throughout the entire range of $\sigma$. In contrast,  AD$^{sup}$ performs  poorly, and has power close to zero when $\sigma<1$.

As we now show, the poor performance of  AD$^{sup}$/HC stems from its specific normalization of the deviations at the extreme indices \(u_{(1)},u_{(2)},\) etc. To this end, recall that under the null, $\pr{u_{(1)} < x } = 1 - (1-x)^n $. Hence, the probability that the first order statistic is smaller than $1/cn \log \log n $, for some constant \(c>0\).  is given by 
\begin{align*}
    \pr{u_{(1)} < \frac{1}{c n \log \log n}} = \frac{1+o(1)}{c \log \log n} .
\end{align*}
For such values of \(u_{(1)}\), the corresponding HC deviation at the first index is
\begin{align*}
    \sqrt{n} \frac{\frac{1}{n} - u_{(1)}}{\sqrt{u_{(1)} (1 - u_{(1)})}} > \sqrt{c\log\log n}(1+o(1)).
\end{align*}
It is now instructive to plug in some specific number into the above equations. In particular, for \(n=100\) samples as in Figure \ref{fig:Gaussian_sim}, a value  $c=65.48$ gives that with probability of $1\%$
the deviation of the first order statistic is at least  $\sqrt{c \log \log n} \approx 10$.

Now suppose we conduct an HC test at a false alarm level of \(\alpha=1\%\). The above calculation has two important implications:\ First, the finite sample threshold of the HC\ test at \(n=100\) must clearly satisfy $t_\alpha>10$.
This value is significantly larger than its asymptotic value of $\sqrt{2 \log \log n }(1+o(1))\approx1.74$ (see Theorem \ref{thm:asymptotic_sup_hatVn}).
Since the decay of \(1/\log\log n\) to zero is extremely slow, the above illustrates the very slow convergence of the  AD$^{sup}$ or HC  distribution to its asymptotic limit. 
Second, such a high threshold prevents detection of significant deviations near the center of the distribution, as indeed is shown empirically in Figure \ref{fig:Gaussian_sim}.
As an  example, a significant deviation from the null of  $u_{(n/2)} = 1/4$ which corresponds to about 5.8 standard deviations  cannot be detected by the HC test at level \(\alpha=1\%\). 

We remark that HC's problematic handling of \(u_{(1)}\) was already noted by \cite{DonohoJin2004}, and  discussed in several recent works \citep{Walther, Finner2014Bernoulli, Finner2015Biometrical, Siegmund_14}.
Finally, we note that in our numerical example, removing \(u_{(1)}\) from the HC test does not resolve the problem, since the next extreme order statistics \(u_{(2)},u_{(3)}\) etc., also have a non-negligible probability to induce very large HC\ values. In contrast, the $M_n$ statistic puts all of these deviations on an equal scale.

\subsection{Detecting Sparse Gaussian Mixtures}
Next, we consider the problem of detecting a sparse Gaussian mixture of the form  \eqref{eq:H0H1_sparse_normal_mixture}, where the parameter 
\(\mu\) is assumed positive. We hence compare the following four one-sided test statistics:  \(\max X_i\),  \(\sum X_i\),   \( \mbox{HC}^{2004} \) and
$M_n^+$. 
\begin{figure}
    \centering
    \makebox{
        \includegraphics[width = 0.48\textwidth]{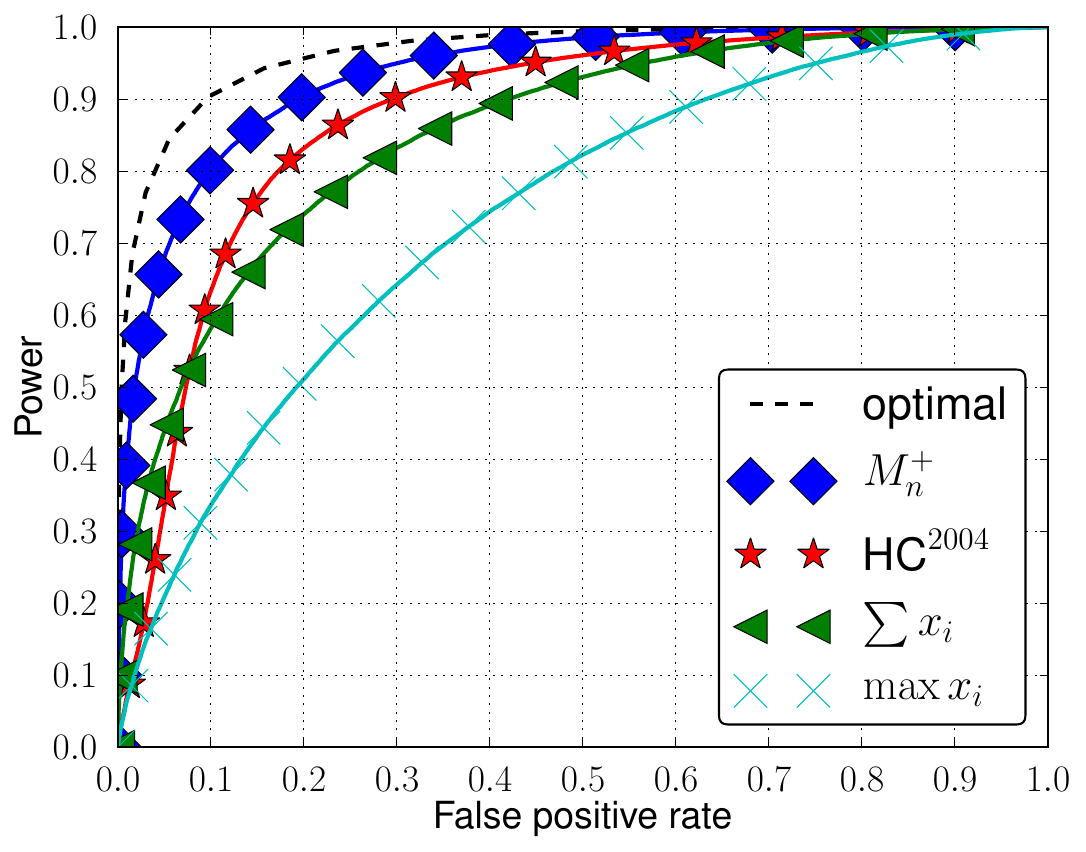}
        \includegraphics[width = 0.48\textwidth]{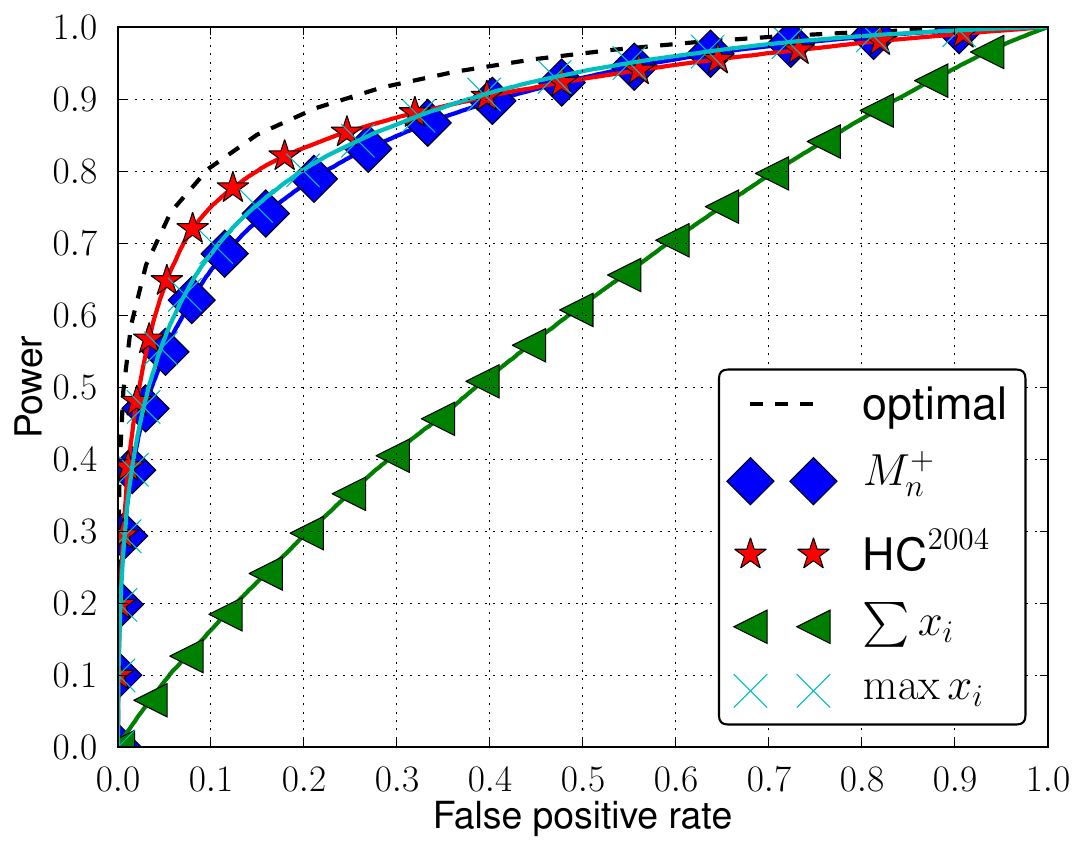}    }
    \caption{\label{fig:rare_weak_roc_curve_n10000}ROC curves for the rare-weak Gaussian mixture model with \(n=10,000\) samples, and with sparsity and contamination levels: \(\epsilon=0.01, \mu=1.5\) (left); \(\epsilon=0.001, \mu=3\) (right).}
\end{figure}

Figure \ref{fig:rare_weak_roc_curve_n10000} compares the resulting Receiver Operating Characteristic (ROC) curves for two choices of \(\epsilon\) and \(\mu\), both with \(n=10,000\) samples.
The optimal curve is that of the likelihood-ratio test,  which unlike the other statistics, is  model specific and requires explicit knowledge of the values of \( \epsilon \) and \( \mu\).

While asymptotically as \(n\to\infty\), both
 \( \mbox{HC}^{2004} \) and $M_n^+$ achieve the same performance as that of the optimal LR test, for finite values of \(n\), as seen in the figure, the gap in detection power may be large. Moreover, for some \((\mu,\epsilon)\) values, HC\(^{2004}\) achieves a higher ROC curve, whereas for others $M_n^+$ is better. A natural question thus follows: For a finite number of samples \(n\), as a function of the two parameters \( \epsilon \) and \( \mu \), which of these four tests has greater power? 
To study this question, we made the following extensive simulation: for many different values of \((\mu,\epsilon)\), we empirically computed the detection power
 of the four tests mentioned above at a significance level of \(\alpha=5\%\),  both for \(n=1000\) and for \(n=10,000\) samples.
For each sparsity value  $\epsilon$ and contamination level $\mu$
we declared that a test \(T_{1}\) was a clear winner if it had a significantly lower misdetection rate, namely if \(\min_{j=2,3,4}\Pr[T_j=\mathcal H_0|\mathcal H_1]/\Pr[T_1=\mathcal H_0|\mathcal H_1] > 1.1$. 

Figure \ref{fig:rareweak_power_sweep} shows the regions in the \((\mu,\epsilon)\) plane where different tests were declared as clear winners. 
First, as the figure shows, at the upper left part in the \((\mu,\epsilon)\) plane, \(\sum X_i\) is the best test statistic. This is expected, since in this region $\epsilon$ is relatively large and leads to a significant shift in the mean of the distribution. At the other extreme, in the lower right part of the \((\mu,\epsilon)\) plane, where $\epsilon$ is small but $\mu$ is large, very few samples are contaminated and here the HC$^{2004}$ test statistic works best, with the max statistic being a close second. In the intermediate region, which would naturally be characterized as the rare/weak region, it is the $M_n^+$ test that has a higher power. Second, while not shown in the plot, we note that the  $M_n^+$ test had similar power to that of the $R_n^+$ test. Finally,
in this simulation the $\mbox{HC}^{2008}$ test performed worse than at least one of the other tests for all values of \((\mu,\epsilon)\). 
\begin{figure}
    \centering
    \makebox{
        \includegraphics[width = 0.48\textwidth]{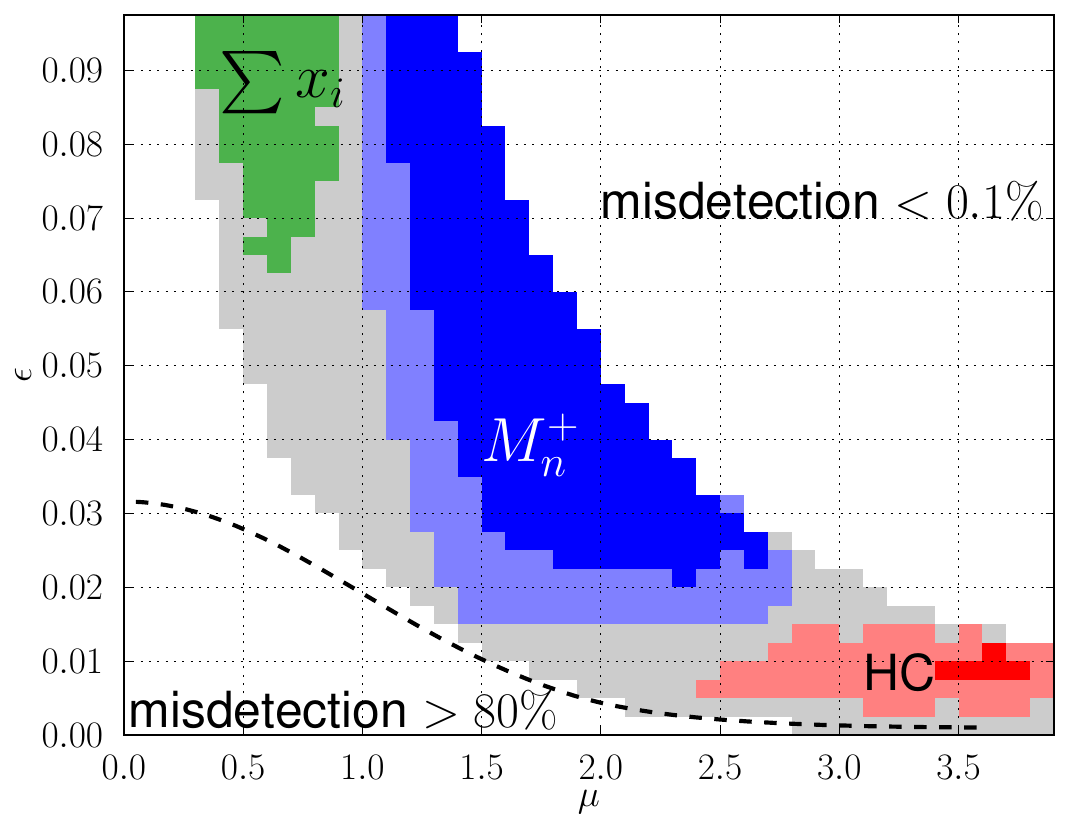}
        \includegraphics[width = 0.48\textwidth]{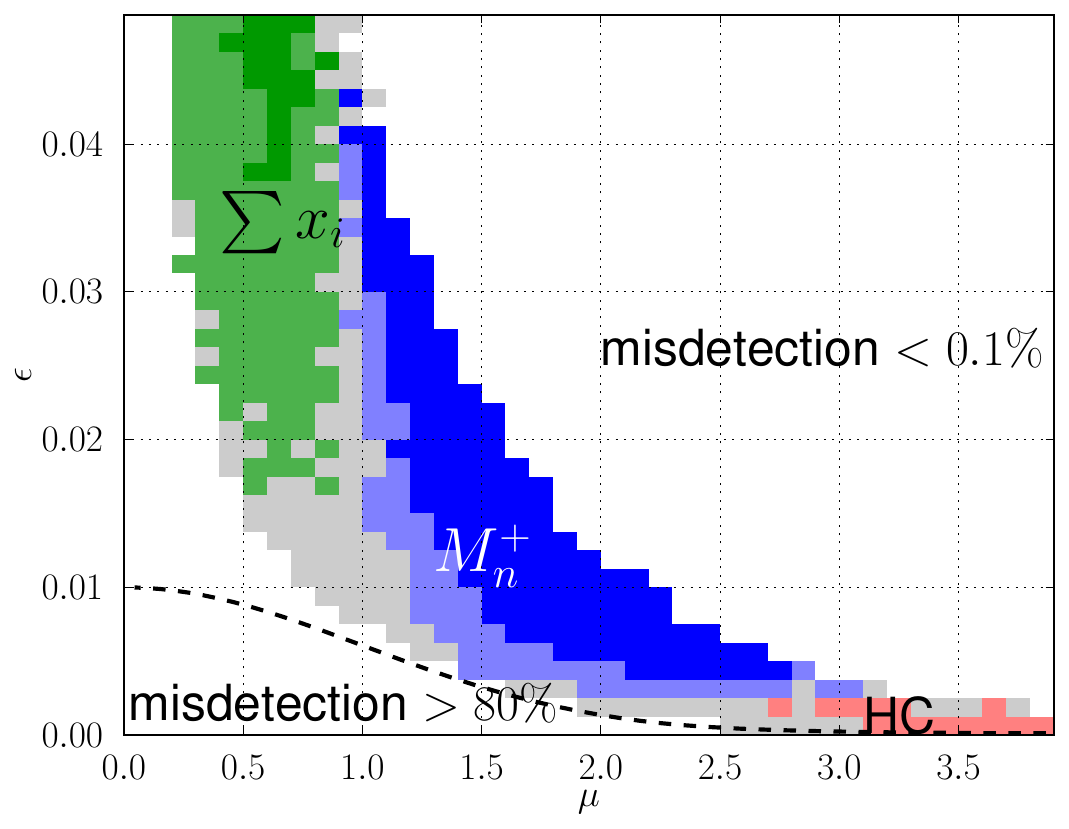}    }
    \caption{\label{fig:rareweak_power_sweep}
        (Best viewed in color) Comparison of tests for detecting rare-weak Gaussian mixtures
        $(1-\epsilon) \mathcal{N}(0,1) + \epsilon \mathcal{N}(\mu, 1)$ vs. $\mathcal{N}(0,1)$.
        Colored blobs represent regions 
        where the misdetection rate of the second-best test
        divided by that of the best test was larger than 1.1.
        The dark centers signify regions where this ratio was larger than 1.5.
        The gray band delineates the zone where misdetection is
        in the range $0.1\% - 80\%$.
        The dotted line is the asymptotic
        detection boundary \eqref{eq:rareweak_asymptotic_detection_boundary}
        when substituting $\epsilon = n^{-\beta}$, $\mu = \sqrt{2 r \log n}$. Left panel: \(n=1,000\); right panel: $n=10,000$. 
    }
\end{figure}

\appendix
\section{Auxiliary Lemmas}
\label{sec:preliminaries}


\subsection{Asymptotics of the Beta distribution}

\comment{
Let $\mu_{\alpha, \beta}$ and $\sigma_{\alpha, \beta}$ denote the mean and standard deviation 
of a Beta$(\alpha,\beta)$ distribution,    
\begin{align*}
    \mu_{\alpha, \beta} = \frac{\alpha}{\alpha+\beta}, \quad
    \sigma_{\alpha, \beta}^2 = \frac{\alpha \beta}{(\alpha+\beta)^2 (\alpha+\beta+1)} \,.
\end{align*}
}
As is well known, when both $\alpha, \beta \rightarrow \infty$, the Beta$(\alpha, \beta)$ distribution approaches $\mathcal{N}(\mu, \sigma^2)$, where $\mu$ and\ $\sigma$ are the mean and standard deviation of the Beta random variable. 
The following lemma quantifies the error in this approximation.
For other approximations, see for example \cite{Peizer_Pratt,Pratt_68}.
\begin{lemma} \label{lemma:beta_approaches_gaussian}
    Let $f_{\alpha, \beta}$ be the density of a $\text{\emph{Beta}}(\alpha, \beta)$ random variable and let 
    $g_t(\alpha, \beta)$ be its value at $t$ standard deviations from the mean. i.e.
    \[
        g_t(\alpha, \beta) = f_{\alpha, \beta}(\mu_{\alpha, \beta} + \sigma_{\alpha, \beta} \cdot t).
    \]
    For any fixed $t$, as both $\alpha, \beta \rightarrow \infty$,
    \begin{align} \label{eq:f_density}
        g_t(\alpha, \beta) =&  \frac{e^{-t^2/2}}{\sqrt{2\pi} \cdot \sigma_{\alpha, \beta}}
        \times \exp \left[ \frac{1}{\sqrt{\alpha+\beta+1}} \left(\sqrt{\frac{\alpha \vphantom{\beta}}{\beta}} - \sqrt{\frac{\beta}{\alpha \vphantom{\beta}}}\right)t  \right] 
 \\
        &\times \exp\left[ O\left(\frac{\beta}{(\alpha+\beta) \alpha} + \frac{\alpha}{(\alpha+\beta) \beta}\right)t^2 \right] \nonumber  \\
        & \times \exp\left[O \left( \frac{1}{\sqrt{\alpha}}+\frac1{\sqrt{\beta}} \right) t^3\right]\times \left(1+O \left( \frac1\alpha+\frac1\beta \right) \right) \nonumber  \,.
    \end{align}
\end{lemma}
\begin{remark}
    For any fixed $t$, as $\alpha, \beta \rightarrow \infty$ all error terms tend to zero,
    hence demonstrating that the distribution of non-extreme order statistics converges to a Gaussian.
    However, for this approximation to be accurate,
    all correction terms must be small, which may require huge sample sizes.   
    As an example, with $t=2$ standard deviations and $\alpha=n^{1/4}$,
    to have $|t^3|/\sqrt{\alpha}<0.1$ 
    we need $n > 1.7 \times 10^{15}$ samples,
    far beyond the reach of almost any scientific study.
\end{remark}
\begin{remark}
    A closer inspection of the proof below shows that Lemma \ref{lemma:beta_approaches_gaussian} continues to hold
even if $t = t(\alpha, \beta) \rightarrow \infty$, provided that $\alpha, \beta \rightarrow \infty$ and
    \begin{equation}
        t\cdot \max
        \left(\sqrt{\tfrac{\alpha \vphantom{\beta}}{(\alpha+\beta)\beta}} ,
        \sqrt{\tfrac{\beta}{(\alpha+\beta)\alpha}} \right) \rightarrow 0 \,.
        \label{eq:t_tend_infty}
    \end{equation}
    This shall prove to be useful later on. 
\end{remark}

\begin{proof}[Proof of Lemma \ref{lemma:beta_approaches_gaussian}]
    For convenience we denote
    \begin{align*}
        A := \alpha-1,\quad B := \beta-1, \quad n := \alpha+\beta-1 = A+B+1 \,.
    \end{align*}
    In terms of these variables, the mean and variance of Beta$(\alpha,\beta)$ are 
    \begin{align*}
        \mu = \frac{A+1}{n+1}, \quad \sigma^2 =\frac{(A+1)(B+1)}{(n+1)^2(n+2)} ,
    \end{align*}
    whereas its density is
    $f(x) = \frac{n!}{A!B!} x^A (1-x)^B$.
    At $x = \mu + \sigma t$, we obtain
    \begin{align}
        \label{eq:beta_density_1}
        f(\mu + \sigma t) &=  \frac{n!}{A!B!}\mu^A(1-\mu)^B  \left(1+\tfrac{\sigma}{\mu}t\right)^A
        \left(1-\tfrac{\sigma}{1-\mu}t\right)^B\,.
    \end{align}
    Using Stirling's approximation, that $n!=\sqrt{2\pi n}(n/e)^n(1+O(1/n))$,
    and the fact that \(\sigma = \sqrt{AB/n^3}(1+O(1/A + 1/B))\), 
    we obtain that as both \(A,B\to\infty\), 
    \begin{align*}
        \frac{n!}{A!B!}\mu^A(1-\mu)^B 
        &= \frac1{\sqrt{2\pi} \cdot \sigma} \left(1+O\left(\frac1A+\frac1B\right)\right).
    \end{align*}
    Next, we write the remaining terms in (\ref{eq:beta_density_1}) , 
    \begin{align}
    \left(1+\tfrac{\sigma}{\mu}t\right)^A
        \left(1-\tfrac{\sigma}{1-\mu}t\right)^B = \exp\left[A\ln(1+\tfrac{\sigma}\mu t)+B\ln(1-\tfrac\sigma{1-\mu}t) \right]\,.
                \label{eq:bla_bla}
    \end{align}
    Note that as $A,B\to\infty$, both $\sigma/\mu$ and $\sigma/(1-\mu)$ tend to zero.
    Hence, for either a fixed \(t\), or   $t=t(\alpha,\beta)$ slowly growing to $\infty$ such that Eq. (\ref{eq:t_tend_infty}) holds,  we may replace the logarithms in Eq. (\ref{eq:bla_bla}) by their Taylor expansion with small approximation errors
    \begin{align}
        \label{eq:beta_density_2}
        \left(1+\tfrac{\sigma}{\mu}t\right)^A \left(1-\tfrac{\sigma}{1-\mu}t\right)^B &= 
        \exp\left[\sigma t \left( \tfrac{A}\mu-\tfrac{B}{1-\mu} \right) -\tfrac{\sigma^2t^2}{2} \left( \tfrac{A}{\mu^2}+\tfrac{B}{(1-\mu)^2} \right)\right]\\
        & \times \exp \left[O \left( \tfrac{A\sigma^3}{\mu^3}+\tfrac{B\sigma^3}{(1-\mu)^3} \right) t^3\right] \nonumber .
    \end{align}
    Simple algebra gives
    \[
        \sigma t\left(\tfrac{A}\mu-\tfrac{B}{1-\mu}\right)
        =
        \tfrac{t}{\sqrt{n+2}} \left( \sqrt{\tfrac{A+1\vphantom{B} }{B+1}}-\sqrt{\tfrac{B+1}{A+1}} \right) \,.
    \]
    Similarly, 
    \[
        \tfrac{\sigma^2t^2}{2} \left( \tfrac{A}{\mu^2}+\tfrac{B}{(1-\mu)^2} \right) = \tfrac{t^2}{2} \left(1+ O \left(\tfrac{B}{nA}
+ \tfrac{A}{nB} \right) \right) \,.
    \]
    Finally, as $A,B\to\infty$, the cubic term in Eq. (\ref{eq:beta_density_2}) is of order $O \left( \tfrac{1}{\sqrt{A}}+\tfrac{1}{\sqrt{B}} \right)t^3$.
    Combining all of these results concludes the proof of the lemma. 
\end{proof}
We present a simple corollary of Lemma \ref{lemma:beta_approaches_gaussian}, which shall prove useful in studying the asymptotic behavior of $M_n$ under the null hypothesis.

\begin{corollary}\label{cor:beta_approaches_gaussian_simplified}
let $\{\alpha_n\}$ be a sequence of numbers converging to infinity. Let
$\mu_n, \sigma_n^2$ and $f_n$ denote the mean, variance and density of a Beta$(\alpha_n, n - \alpha_n + 1)$ distribution, respectively.
Furthermore, let $g(n)$ be any positive function satisfying $g(n)=o(\min\{\alpha_n,n-\alpha_n\})$.
    Then, for large values of $n$ we have a lower bound on the density near the mean, 
    \begin{align}
        \label{eq:f_density_simplified}
        f_n(\mu_n + \sigma_n \cdot t) \ge \frac{e^{-t^2/2}}{\sqrt{2\pi} \cdot
\sigma_n}  \left(1 - \frac{t^3}{\sqrt{g(n)}} - \frac{1}{g(n)}  \right) \,. 
    \end{align}
\end{corollary}
\begin{proof}
    Follows from an inspection of the various error terms in Eq. \eqref{eq:f_density}.
\end{proof}

\subsection{Supremum of the Standardized Empirical Process}
\label{sec:sup_empirical_process}
The standardized empirical process plays a central role in our analysis of the $M_n$ statistic.
We begin with its definition followed by several known results regarding
the magnitude and location of its supremum.
 
\begin{definition}
    \label{def:empirical_processes}
    Let $X_1, \ldots ,X_n$ be i.i.d. random variables from some continuous distribution $F$
    and let $\hat F_n(x)=\frac1n\sum_i \mathbf{1}(X_i \le x)$ denote their empirical cdf.
    The normalized empirical process is defined as
    \begin{equation}
        \label{eq:normalized_empirical_process}
        V_n(x) = \sqrt{n}\frac{\hat F_n(x)-F(x)}{\sqrt{F(x)(1-F(x))}} \quad \mbox{
for  }\  0 < F(x) < 1 \,.
    \end{equation}
    Similarly, the standardized empirical process is
    \begin{equation}
        \label{eq:standardized_empirical_process}
        \hat{V}_n(x) = \sqrt{n}\frac{\hat F_n(x)-F(x)}{\sqrt{\hat F_n(x)(1-\hat F_n(x))}} \quad \mbox{ for  }\  0 < F_n(x) < 1 \,.
    \end{equation}
\end{definition}
Of particular interest to us is the supremum of $\hat{V}_n$.
The following lemma provides an equivalent expression for this quantity.
\begin{lemma}
    Let $U_1, \ldots, U_n \simiid U[0,1]$ and let
    $\hat{V}_n(u)$ be the standardized empirical process of Eq. \eqref{eq:standardized_empirical_process}.
    Then
    \begin{equation}
        \label{eq:sup_standardized_process_equivalent}
        \sup_{U_{(1)} < u < U_{(n)}} \hat{V}_n(u) = \max_{i=1, \ldots, n-1}\sqrt{n}\frac{\frac{i}{n} - U_{(i)}}{\sqrt{\frac{i}{n}(1-\frac{i}{n})}} \,.
    \end{equation}
\end{lemma}
\begin{proof}
    Without loss of generality, we may assume that $F = U[0,1]$.
    For any $0 < c < 1$, since $(c-u)/\sqrt{c(1-c)}$
    is monotone decreasing in  \(u\), the supremum in the left hand side of Eq. \eqref{eq:sup_standardized_process_equivalent}
    is attained at the left edge of one of the intervals of the piecewise-constant function
    $\hat{F}_n$. Hence,
    \begin{align*}
        \sup_{U_{(1)} < u < U_{(n)}} \hat{V}_n(u) 
        &= \max_{u \in \{U_{(1)}, \ldots, U_{(n-1)} \}} \sqrt{n}\frac{\hat{F}_n(u) - u}{\sqrt{\hat{F}_n(u)(1-\hat{F}_n(u))}} \,.
    \end{align*}
    Since $\hat{F}_n(U_{(i)}) = i/n$, Eq. \eqref{eq:sup_standardized_process_equivalent} follows.
\end{proof}
We recall Theorem 1 of \cite{Eicker1979}, which gives the asymptotic distribution of the supremum of $\hat{V}_n(u)$, see also \cite{Csorgo1986}.
\begin{thm} \label{thm:asymptotic_sup_hatVn}
    Let $U_1, \ldots, U_n \simiid U[0,1].$
    As $\ntoinf$, 
    
    \begin{align*}
        \pr{\sup_{U_{(1)} < u < U_{(n)}} \hat{V}_n(u) < \sqrt{2\log \log n} + \frac{\log \log \log n}{2 \sqrt{2\log \log n}} + \frac{1}{\sqrt{2 \log \log n}} \cdot t} \rightarrow e^{-e^{-t}/\sqrt{\pi} }.
    \end{align*}
\end{thm}

Furthermore, the next lemma, which follows from the main Theorem of \cite{Jaeschke1979},
implies that this supremum is rarely attained at one of the extreme order statistics.
\begin{lemma} \label{lemma:asymptotic_standardized_deviation_location}
    Let $k>0$ and let $I$ be the union of  intervals containing the first and last \(\log^k n\) order statistics, $I = (U_{(1)}, U_{(\log^k n)}] \cup [U_{(n - \log^k n)}, U_{(n)})$. Then
    \begin{align}
        \pr{\sup_{u \in I} \hat{V}_n(u) < \sup_{U_{(1)} < u < U_{(n)}} \hat{V}_n(u)} \xrightarrow{n \rightarrow \infty} 1 \,.
    \end{align}
\end{lemma}

\section{Proofs of Theorems} \label{sec:proofs}

\subsection{Proof of Theorem \ref{thm:Mnplus_asymptotic_distribution}}
We are now ready to bound the distribution of $M_n$ under the null.
We start with several technical lemmas.
\begin{lemma} \label{lemma:tau_approx}    
    Let $\mu$ and $\sigma^2$ denote the mean and variance of a Beta$(i,n-i+1)$ random variable. For any $x \in [0,1]$,
    \[
        \frac{\mu - x}{\sigma} = \sqrt{n}\frac{\frac{i}{n}-x}{\sqrt{\frac{i}{n}(1-\frac{i}{n})}} \left( 1+ O\left( \frac{1}{n} \right)\right).
    \]
\end{lemma}
\begin{proof}
    Follows by straightforward algebraic manipulations.
\end{proof}

\begin{lemma} \label{lemma:normal_tail_tau}
    Let
    \[
        \tau(t) := \sqrt{2\log \log n} + \frac{\log \log \log n}{2 \sqrt{2\log \log
n}} + \frac{1}{\sqrt{2 \log \log n}} \cdot t .
    \]
    Taking $n \to \infty$, if $t = o(\sqrt{\log \log n})$ then
    \begin{align*}
        \frac{1}{\tau(t)} e^{-\tau(t)^2/2} = \frac{1+o(1)}{\sqrt{2} \log n \log \log n} e^{-t}.
    \end{align*}
\end{lemma}
\begin{proof} For $t=o(\sqrt{\log\log n})$, we have that
    \begin{align*}
        \tau^2(t) &= 2 \log \log n + \frac{(\log \log \log n)^2}{8 \log \log n} + \frac{t^2}{2\log \log n} + \log \log \log n + 2t + \frac{\log \log \log n}{2 \log \log n} \cdot t \\
        &=
        2 \log \log n + \log \log \log n + 2t + o(1).
    \end{align*}
    Therefore,
    \begin{align*}
        \frac{1}{\tau(t)} \cdot e^{-\tau^2(t)/2} &= \frac{1}{\sqrt{2 \log \log n}+ o(1) } \cdot \frac{\left( 1+o(1) \right)e^{-t}}{\log n \sqrt{ \log \log n}}
        =
        \frac{1+o(1)}{\sqrt{2} \log n \log \log n} e^{-t}.
    \end{align*}
\end{proof}

\begin{lemma} \label{lemma:nearly_gaussian_density_tail}
    Let $\epsilon > 0$ and $a > 0$ be constants.
    Let $f:\mathbb{R} \to \mathbb{R}$ be a non-negative function that satisfies the following conditions,
    \begin{enumerate}
        \item $\int f(x) dx = 1$.
        \item $ f(x) \ge \frac{1}{\sqrt{2\pi}}e^{-x^2/2} \left(1-\epsilon \right)$ in the range $x \in [-a,a]$.
    \end{enumerate}
    Then for any $t \in [0, a]$,
    \begin{align*}
        \int_{-\infty}^{-t} f(x)dx
        \le
        \frac{1}{\sqrt{2\pi}}\frac{1}{t}e^{-t^2/2} + \frac{1}{\sqrt{2\pi}}\frac{1}{a}e^{-a^2/2} +\epsilon.
    \end{align*}
\end{lemma}
\begin{proof}
    \begin{align} \label{eq:p_i_star_bound}        
        \int_{-\infty}^{-t} f(x)dx \,
        \le
        1 - \int_{-t}^{a} f(x) dx
        \le
        1- (1-\epsilon) \int_{-t}^{a} \frac{1}{\sqrt{2\pi}}e^{-\frac{1}{2}x^2} dx.
    \end{align}
    A simple bound on the Gaussian tail is given by
    \begin{align} \label{eq:gaussian_tail_approximation}
        \int_t^{\infty} e^{-x^2/2} dx
        \le
        \int_t^{\infty} \frac{x}{t} e^{-x^2/2} dx
        =
        \frac{1}{t} e^{-t^2/2}.
    \end{align}
    Therefore,
    \begin{align*}
        \int_{-t}^a \frac{1}{\sqrt{2\pi}} e^{-x^2/2}dx &= 1 - \int_a^\infty \frac{1}{\sqrt{2\pi}}e^{-x^2/2}dx - \int_{-\infty}^{-t}
\frac{1}{\sqrt{2\pi}} e^{-x^2/2} dx \\
        &\ge 1 - \frac{1}{\sqrt{2\pi}}\frac{e^{-a^2/2}}{a} - \frac{1}{\sqrt{2\pi}} \frac{e^{-t^2/2}}{t}.
    \end{align*}
   
    Plugging this bound into Eq. \eqref{eq:p_i_star_bound} finishes the proof.    
\end{proof}
Now we are ready to combine the lemmas given above to produce a bound on the asymptotic distribution of $M_n^+$.
\begin{lemma} \label{thm:Mn_upper_bound}
    For any fixed $x > 0$ and $\epsilon>0$,
    \[
        \pr{M_n^+ \ge \frac{x}{2\log n \log \log n} \bigg| \mathcal{H}_0} \le e^{-x} + o(1).
    \]
\end{lemma}

\begin{proof}
    Following Section \ref{sec:order_statistics}, we study
    w.l.o.g. the distribution of $M_n^+(U_1, \ldots, U_n)$, where $U_i \simiid U[0,1]$. The main idea is to look at the index $i_{*} \in \{1,2,\ldots,n\}$
    where the standardized empirical process of \eqref{eq:standardized_empirical_process}
    attains its maximum, i.e.
    \begin{align}
        \label{eq:u_i_star_def}
        \sup_{U_{(1)} < u <\ U_{(n)}} \hat{V}_n(u) = \hat{V}_n(U_{(i_*)}) = \sqrt{n}\frac{\frac{i_*}{n}-U_{(i_*)}}{\sqrt{\frac{i_*}{n}(1-\frac{i_*}{n})}} ,
    \end{align}
    and infer a bound on $p_{(i_*)}$.
    Denote by $f_*(x)$ the density of a Beta$(i_*, n-i_*+1)$ random variable,
    by definition \eqref{def:p_i},
    \begin{align*}
        p_{(i_*)} := \int_0^{U_{(i_*)}} f_*(x) dx.
    \end{align*}    
    Let $\mu_*$ and $\sigma_*^2$ denote the mean and variance of a Beta$(i_*,n-i_*+1)$ random variable and let
    \[
        \tau_* := \frac{\mu_* - U_{(i_*)}}{\sigma_*}
    \]
    be the z-score of $U_{(i_*)}$.
   With a change of variables $t = (x-\mu_*)/\sigma_*$, we obtain
    \begin{align*}
        p_{(i_*)} = \int_{-\mu_*/\sigma_*}^{-\tau_*} f_*(\mu_* + \sigma_* t) \cdot \sigma_* dt \le \int_{-\infty}^{\tau_*} f_*(\mu_* + \sigma_* t) \cdot \sigma_* dt .
    \end{align*}
        By Lemma \ref{lemma:tau_approx},
    \begin{align} \label{eq:tau_star_approx_Vnhat}
        \tau_*  = \hat{V}_n(U_{(i_*)}) (1+O(1/n)) 
    \end{align}
    For any $\epsilon_1 > 0$, the following statements hold with probability $> 1 - \epsilon_1$ for large values of $n$:
    \begin{enumerate}
        \item $\log^{12} n < i_* < n - \log^{12} n$ (from Lemma \ref{lemma:asymptotic_standardized_deviation_location}).
        \item $\tau_* < \log n$ (from Theorem \ref{thm:asymptotic_sup_hatVn} and Eq. \eqref{eq:tau_star_approx_Vnhat})
    \end{enumerate}
    By Corollary \ref{cor:beta_approaches_gaussian_simplified},
    for any $t \in [-\log n, \log n]$, if $n$ is large enough, then
    \begin{align*}
        f_*(\mu_* + \sigma_* \cdot t) \ge \frac{e^{-t^2/2}}{\sqrt{2\pi} \cdot
\sigma_*}  \left(1 - \frac{1}{\log^2 n}  \right)
    \end{align*}
    Hence, we may apply Lemma \ref{lemma:nearly_gaussian_density_tail} with $a = \log n$ and obtain that the following holds with high probability,
    \begin{align} \label{eq:pistarbound}
         p_{(i_*)}
         &\le
         \frac{1}{\sqrt{2\pi}} \frac{1}{\tau_*} e^{-\frac12 \tau_*^2}
         +
         \frac{1}{\sqrt{2\pi}} \frac{1}{\log n} e^{- \log^2 n / 2}
         +
         \frac{1}{\log^2 n} \\
         &\le
         \frac{1}{\sqrt{2\pi}} \frac{1}{\tau_*} e^{-\frac12 \tau_*^2}
         +
         \frac{2}{\log^2 n}. \nonumber     
    \end{align}
    Now, let
    \[
        \tau(t) := \sqrt{2\log \log n} + \frac{\log \log \log n}{2 \sqrt{2\log \log n}} + \frac{1}{\sqrt{2 \log \log n}} \cdot t.
    \]
    By Theorem \ref{thm:asymptotic_sup_hatVn} and Eq. \eqref{eq:tau_star_approx_Vnhat}, for any fixed $t$ and any $\epsilon_2 > 0$
    \begin{align*}
        \pr{\tau_* > \tau(t)} > 1 - e^{-e^{-t}/\sqrt{\pi}} - \epsilon_2.
    \end{align*}
    Since the function $\frac{1}{x} e^{-\frac12 x^2}$ is monotone decreasing,
    \[
        \tau_* > \tau(t) \quad \Longleftrightarrow \quad \frac{1}{\tau_*}e^{-\frac12 \tau_*^2} < \frac{1}{\tau(t)}e^{-\frac12 \tau^2(t)}.
    \]
    By Lemma \ref{lemma:normal_tail_tau},
    \[
        \frac{1}{\tau(t)} e^{-\tfrac12 \tau(t)^2} = \frac{1+o(1)}{\sqrt{2} \log n \log \log n} e^{-t} 
    \]
    Therefore
    \begin{align} \label{eq:something}
        \pr{\frac{1}{\tau_*}e^{-\frac12 \tau_*^2} <  \frac{1+\epsilon/2}{\sqrt{2} \log n \log \log n} e^{-t}} 
        \ge
        \pr{\tau_* > \tau(t)} 
        >
        1 - e^{-e^{-t}/\sqrt{\pi}} - \epsilon_2.
    \end{align}
    Finally,
    \begin{align*}
        &\pr{M_n^+ \ge \frac{1+\epsilon}{2\sqrt{\pi} \log n \log \log n} e^{-t}} \\
        &\le
        \pr{p_{(i_*)} \ge \frac{1+\epsilon}{2\sqrt{\pi} \log n \log \log n} e^{-t}} && (\text{by definition, } M_n^+ \le p_{(i_*)} )\\
        &<
        \pr{\frac{1}{\sqrt{2\pi}} \frac{1}{\tau_*} e^{-\frac12 \tau_*^2} + \frac{2}{\log^2 n} \ge \frac{1+\epsilon}{2\sqrt{\pi} \log n \log \log n} e^{-t}} + \epsilon_1 && \text{(by Eq. \ref{eq:pistarbound})} \\
        &=
        \pr{\frac{1}{\sqrt{2\pi}} \frac{1}{\tau_*} e^{-\frac12 \tau_*^2} \ge \frac{1+\epsilon}{2\sqrt{\pi} \log n \log \log n} e^{-t} -  \frac{2}{\log^2 n}} + \epsilon_1\\
        &<
        \pr{\frac{1}{\sqrt{2\pi}\tau_*} e^{-\frac12 \tau_*^2} \ge \frac{1+\epsilon/2}{ 2\sqrt{\pi}\log n \log \log n} e^{-t}} + \epsilon_1 && \text{(since $t$ is fixed and $n \to \infty$)} \\
        &<
        e^{-e^{-t}/\sqrt{\pi}} + \epsilon_2 + \epsilon_1. && \text{(by Eq. \eqref{eq:something})}
    \end{align*}
    This claim is true for any $\epsilon_1, \epsilon_2>0$, therefore
         \begin{align*}
        \pr{M_n^+ \ge \frac{1+\epsilon}{2\sqrt{\pi} \log n \log \log n} e^{-t}}
        \le
        e^{-e^{-t}/\sqrt{\pi}}.
    \end{align*}
    Choosing $t = -\log \tfrac{\sqrt{\pi} x}{1+\epsilon}$ and taking $\epsilon \to 0$ finishes the proof.
\end{proof}
Next, we consider the location where, under the null, the $M_n^+$ statistic attains its minimal value.
Lemma \ref{lemma:asymptotic_standardized_deviation_location} shows that for the standardized empirical process,
the probability of the supremum being attained at one of the extreme indices approaches zero as $n \rightarrow \infty$. We prove a similar result regarding $M_n, M_n^+, M_n^-$.

\begin{lemma} \label{lemma:argmax_Mn_plus}
    Let $i_*$ denote the location of the most statistically significant deviation as measured by the one-sided $M_n^+$ statistic,
    \(
        i_* = \argmin_{1 \le i \le n} p_{(i)}.
    \)
    Then under the null hypothesis,
    \[
        \pr{\log n \le i_* \le n-\log n \big | \mathcal{H}_0} \xrightarrow{n \rightarrow \infty} 1
    \]
    and the same result holds for the $M_n^-$ and $M_n$ statistics.
\end{lemma}
\begin{proof}
        Denote $C_n := \log \log \log n / \log n \log \log n$. Clearly,
        \begin{align*}
        \pr{i_* \le \log n \big| \mathcal{H}_0 }
        &= 
        \pr{i_* \le \log n \text{ and } M_n^+ < C_n \big| \mathcal{H}_0 }
        \\
        &+\pr{i_* \le \log n \text{ and } M_n^+ \ge C_n \big| \mathcal{H}_0}.        
    \end{align*}
    By Lemma \ref{thm:Mn_upper_bound}, the second summand vanishes as $n \to \infty$.
    Recall that $p_{(i)}$ is a $p$-value, hence uniformly distributed under the null hypothesis.
    Using this fact and a union bound, we obtain that the first summand vanishes as well,
since    \begin{align*}
        &\pr{i_*^+ \le \log n \text{ and } M_n^+ < C_n \big| \mathcal{H}_0 }
        \le
        \pr{\exists i \le \log n: p_{(i)} < C_n \big| \mathcal{H}_0} \\
        &\le
        \sum_{i=1}^{\lfloor \log n \rfloor} \pr{p_{(i)} < C_n \big| \mathcal{H}_0}
        =
        \lfloor \log n \rfloor \cdot C_n\xrightarrow{n \to \infty} 0.
    \end{align*}
    The same argument works for the last $\log n$ elements and also for the $M_n^-$ statistic.
    By implication it holds for the $M_n$ statistic as well.
\end{proof}
\begin{lemma} \label{lemma:Mnplus_lower_bound_tau}
    Let $U_1, \ldots, U_n \simiid U[0,1]$
    and let $i_* = \argmin_{1 \le i \le n} p_{(i)}$ be the index where $M_n^+(U_1, \ldots, U_n)$ attains its value.
    Denote the mean and variance of a Beta$(i_*,n-i_*+1)$ random variable by $\mu$ and $\sigma^2$ respectively
    and denote the z-score of $U_{(i_*)}$ by $\tau := (\mu - U_{(i_*)})/\sigma$, then for every fixed $\epsilon > 0$,
    \[
        \pr{M_n^+  >  \frac{1-\epsilon}{\sqrt{2 \pi}\tau}e^{-\tau^2/2}} \xrightarrow{n \to \infty} 1.
    \]
\end{lemma}

\begin{proof}
    By definition
    \begin{align}
        \label{eq:p_i_star_0}
        p_{(i_*)} &:= \int_0^{U_{(i_*)}} f(x) dx = \int_{-\mu/\sigma}^{-\tau} f(\mu + \sigma t) \cdot \sigma dt \,,
    \end{align}
    From Lemma \ref{lemma:argmax_Mn_plus} follows that with high probability $\log n < i_* < n - \log n$.
    Hence $\mu/\sigma > \sqrt{\log n}$ and in particular 
    \begin{align*}
        p_{(i_*)} > \int_{-\sqrt{4\log \log n}}^{-\tau} f(\mu + \sigma t) \cdot \sigma dt \,.
    \end{align*}
    In this domain of integration $|t| < \sqrt{4\log \log n}$, so Corollary \ref{cor:beta_approaches_gaussian_simplified} gives that
    \begin{align}
        p_{(i_*)} &> \frac{1}{\sqrt{2 \pi}} \cdot \int_{-\sqrt{4\log \log n}}^{-\tau} e^{-\frac{1}{2} t^2} dt  \cdot \left( 1 - \frac{1}{\log^{1/3} n} \right)\nonumber \\
            &= \frac{1}{\sqrt{2 \pi}} \left( \int_{-\infty}^{-\tau} e^{-\frac{1}{2} t^2} dt - \int_{-\infty}^{-\sqrt{4\log\log n}} e^{-\frac{1}{2} t^2} dt \right)\cdot \left( 1 - \frac{1}{\log^{1/3} n} \right) \,.
        \label{eq:p_i_star_lower_bound}
    \end{align}
    From Theorem \ref{thm:asymptotic_sup_hatVn} it follows that with high probability $\tau < \sqrt{3 \log \log n}$.
    Therefore, by applying the tail approximation (\ref{eq:gaussian_tail_approximation}) to the last result,
    the second integral becomes negligible with respect to the first, thus for every $\epsilon>0$,
    \begin{align}
        p_{(i_*)} > \frac{1-\epsilon}{\sqrt{2 \pi} \cdot \tau}  e^{-\frac{1}{2} \tau^2}.
    \end{align}
\end{proof}
Now we are ready to finish the proof of Theorem \ref{thm:Mnplus_asymptotic_distribution}
by proving the complementary to Lemma \ref{thm:Mn_upper_bound}.
\begin{lemma}
    For any fixed $x > 0$, as $n \to \infty$,
    \[
        \pr{M_n^+ \ge \frac{x}{2\log n \log \log n} \bigg| \mathcal{H}_0} \ge e^{- x} +\ o(1).
    \]
\end{lemma}
\begin{proof} \label{lemma:Mn_lower_bound}
    Parameterize
    \[
        \tau(t) := \sqrt{2 \log \log n} + \frac{\log \log \log n}{2 \sqrt{2 \log \log n}} +\ \frac{t}{\sqrt{2 \log \log n}}
    \]
    Fix $\epsilon_1, \epsilon_2, \epsilon_3 > 0$. There is some $N(\epsilon_1, \epsilon_2, \epsilon_3)$, such that for every $n > N$,
    \begin{align*}
        & \pr{M_n^+ \ge \frac{x}{2\log n \log \log n}} \\
        &>
        \pr{\frac{1-\epsilon_1}{\sqrt{2 \pi}\tau(t)} e^{-\tau(t)^2/2} \ge \frac{x}{2\log n \log \log n}} -\epsilon_2  && (\text{By Lemma \ref{lemma:Mnplus_lower_bound_tau}})\\
        &>
        \pr{\frac{1-2\epsilon_1}{\sqrt{2 \pi}} \frac{e^{-t}}{\sqrt{2} \log n \log \log n}\ge \frac{x}{2\log n \log \log n}} - \epsilon_2 - \epsilon_3 && \text{(by Lemma \ref{lemma:normal_tail_tau})} \\
        &=
         \pr{e^{-t} \ge \frac{ \sqrt{\pi}x}{1 -2 \epsilon_1} } - \epsilon_2 - \epsilon_3 \\
         &=
         \pr{t < - \log \frac{\sqrt{\pi}x}{1 -2 \epsilon_1}} - \epsilon_2 - \epsilon_3 \xrightarrow{n \to \infty} e^{-x/(1-2\epsilon_1)} - \epsilon_2 - \epsilon_3 && \text{(by Theorem \ref{thm:asymptotic_sup_hatVn})}
    \end{align*}
    Taking $\epsilon_1, \epsilon_2, \epsilon_3 \to 0$ finishes the proof.
\end{proof}
\noindent Combining Lemmas \ref{thm:Mn_upper_bound} and \ref{lemma:Mn_lower_bound} gives the asymptotic null distribution of $M_n^+$.
\begin{corollary} 
    For any fixed $x > 0$,
    \begin{align} \label{eq:Mnplus_asymptotic_distribution}
        \pr{M_n^+ < \frac{x}{2\log n \log \log n} \bigg| \mathcal{H}_0} \xrightarrow{n \to \infty} 1 - e^{- x}.
    \end{align}
\end{corollary}
\noindent The same claim is true for $M_n^-$, by the following lemma.
\begin{lemma} \label{lemma:Mnplus_Mnminus_equally_distributed}
    The null distributions of $M_n^+$ and $M_n^-$ are identical.
\end{lemma}
\begin{proof}
    Assume w.l.o.g. that  $U[0,1]$ is the null distribution and let
    $ U_1, \ldots, U_n $ be distributed according to the null.
    It is easy to show that
    \(
        M_n^+(U_1, \ldots, U_n) = M_n^-(1-U_1, \ldots, 1-U_n).
    \)
    The claim follows from the fact that the vectors $(U_1, \ldots, U_n)$ and $(1-U_1, \ldots, 1-U_n)$ have the same distribution.
\end{proof}
Finally, the derivation of the asymptotic null distribution of $M_n$
is almost identical to that of $M_n^+$. The only difference is that  instead of basing the proof on the distribution of $\hat{V}_n$,
it is based on the distribution of $|\hat{V}_n|$ given in \citep[Theorem 3]{Eicker1979}.
$\hfill\Box$
\subsection{Proof of Theorem \ref{thm:Mn_H_1}}
    Let $t_0 \in \mathbb{R}$ be some point that satisfies
    \[
        |G(t_0) - F(t_0)| = \|G-F\|_\infty \,.
    \]
    Without loss of generality, we assume that $F(t_0) < G(t_0)$ and derive an upper bound on
    $M_n^+$ (in the opposite case the  same upper bound would be obtained on $M_n^-$).
    Let $i_*$ denote the number of random variables $X_i$ smaller than $t_0$.
    Since for all $i$, $\pr{X_i < t} = G(t)$, the random variable $i_*$ follows a binomial distribution,
    \[
        i_* \sim Binomial(n, G(t_0)) \,.
    \]
    Since $F(t_0) < G(t_0)$, for any fixed $0 < \lambda < 1$,
    \begin{align*}
        \pr{\frac{i_*}{n+1} > \lambda G(t_0) + (1-\lambda)F(t_0) } \xrightarrow{n \rightarrow \infty} 1  \,.
    \end{align*}
    This implies that with probability tending to one, 
    \begin{align}
        \frac{i_*}{n+1} - F(t_0)  > \lambda \left( G(t_0) - F(t_0) \right) .
        \label{eq:istar_bound}
    \end{align}
    We show that this implies Eq. \eqref{eq:Mn_H_1} of the theorem.
    To this end, recall that
    \begin{align*}
        M_{n}^+ &\le p_{(i_*)} = \pr{\textrm{Beta}(i_*, n - i_* + 1) < U_{(i_*)}} \,.
    \end{align*}
    By definition, $X_{(i_*)} < t_0$ and therefore $U_{(i_*)} := F(X_{(i_*)}) < F(t_0)$. Thus
    \begin{align}
        M_n^+ &< \pr{\textrm{Beta}(i_*, n - i_* + 1) < F(t_0)} \nonumber \\
        &= \pr{\frac{i_*}{n+1} - \textrm{Beta}(i_*, n - i_* + 1) > \frac{i_*}{n+1} - F(t_0) }  \label{ineq:Mn_at_istar} \,.
    \end{align}
    Hence, from \eqref{eq:istar_bound} follows that with probability tending to one,
    \begin{align*}
        M_n^+ &< \pr{\frac{i_*}{n+1} - \textrm{Beta}(i_*, n - i_* + 1) >  \lambda \left( G(t_0) - F(t_0) \right)} \\
        &< \pr{\left|\frac{i_*}{n+1} -\textrm{Beta}(i_*, n - i_* + 1) \right| >  \lambda \left( G(t_0) - F(t_0) \right)}.
    \end{align*}
    Recall that the expectation of $\textrm{Beta}(i_*, n - i_* + 1)$ is $i_*/(n+1)$
    and that for any $1 \le i_* \le n$ its standard deviation is smaller than $1 / 2\sqrt{n}$.
    Therefore, by Chebyshev's inequality
    \begin{align*}
        M_n^+  <  \frac{1}{4n\lambda^2\| G-F \|_\infty^2} \, .
    \end{align*}
    Setting $\lambda = 1/\sqrt{1+\epsilon}$ concludes the proof. $\hfill \Box$
\subsection{Proof of Theorem \ref{thm:Mn_adaptive_optimality}}

We give a sketch, based on the proof of Theorem 4 by \cite{CaiWu2014}.
In their proof they show that under the alternative, there exists a fixed $0<s<1$ such that
\begin{align}
    \pr{V_{n}(\sqrt{2 s \log n}) > \sqrt{(2+\delta) \log \log n}} \xrightarrow{n \rightarrow \infty} 1\,,
                    \label{eq:V_n_Cai_Wu}
\end{align}
where $V_n$ is the normalized empirical process of Eq. \eqref{eq:normalized_empirical_process}.

Next, let $i_*$ be the (random) number of observations above $\sqrt{2s\log n}$, namely 
$i_* = |\{j\,| X_j > \sqrt{2s \log n}\}|$. Then, from Eq. (\ref{eq:V_n_Cai_Wu}) it follows that
\begin{align}
    \label{eq:cai_wu_Ui}
    \sqrt{n} \frac{i_*/n - U_{(i_*)}}{\sqrt{U_{(i_*)}(1-U_{(i_*)})}} > \sqrt{(2+\delta) \log \log n}\,. \quad\mbox{ (w.h.p.)}
\end{align}
Now, similarly to the proof of Lemma \ref{thm:Mn_upper_bound},
define $\tau := (\mu - U_{(i_*)})/\sigma$. 
Then the $p$-value of $U_{(i_*)}$ may be approximated by the following Gaussian tail,
\begin{align*}
    p_{(i_*)}
    =
    \int_{\tau}^\infty \frac{1+o(1)}{\sqrt{2\pi}}e^{-\frac{1}{2}x^2}dx
    =
    \frac{1+o(1)}{\sqrt{2\pi}} \frac{1}{\tau} e^{-\frac{1}{2}\tau^2} \,.
\end{align*}
From \eqref{eq:cai_wu_Ui} follows that $\tau > \sqrt{(2+\delta) \log \log n}(1+O(1/n))$, and therefore
\begin{align*}
    \pr{M_n \le \ p_{(i_*)} < \frac{1+o(1)}{\sqrt{(2+\delta) \log \log n}\left(\log n \right)^{1+\delta/2}}} \xrightarrow{\ntoinf} 1 \,.
\end{align*}
Setting $\epsilon = \delta/2$  concludes the proof. $\hfill\Box$

\subsection{Proof of Theorem \ref{thm:two_sided_p_value}}

By Lemma \ref{lemma:Mnplus_Mnminus_equally_distributed}, under the null the distributions of $M_n^+$ and $M_n^-$ are equal.
The right inequality in Eq. \eqref{eq:Mn2_bound} follows directly from the union bound
\begin{align*}
    \Pr[M_n \leq c\ |\ \mathcal{H}_0] \le \Pr[M_n^+ \leq c\ |\ \mathcal{H}_0] + \Pr[M_n^- \leq c\ |\ \mathcal{H}_0] = 2 q_c \,.
\end{align*}    
We now prove the left inequality in Eq. \eqref{eq:Mn2_bound}. Let $U_1, \ldots, U_n \simiid U[0,1]$
and denote the joint density of their order statistics by ${\bf U} = (U_{(1)}, \ldots, U_{(n)})$.
Denote the events $M_n^+ > c$ and $M_n^- > c$ by $A$ and $B$ respectively.
According to proposition 3.11 from \citet{KarlinRinott1980}, the random vector $\bf U$ is multivariate totally positive of order 2. 
It is easy to show that the indicator functions $\mathbf{1}_A({\bf U})$ and $\mathbf{1}_{B^c}({\bf U})$ are
monotone-increasing in $\mathbb{R}^n$, hence by \citet[Theorem 4.2]{KarlinRinott1980}, we have
\[
    \e{\mathbf{1}_A(\X) \mathbf{1}_{B^c}(\X)} \ge \e{\mathbf{1}_A(\X)}\e{\mathbf{1}_{B^c}(\X)} .
\]
Equivalently, \( \pr{A \wedge B^c} \ge \pr{A}\cdot  \pr{B^c} \).
Therefore,
\begin{align*}
    \pr{M_n > c} &= \pr{A \wedge B}  = \pr{A} - \pr{A \wedge B^c} \\
    &\le \pr{A}-\pr{A}\pr{B^c} = \pr{A}\pr{B} = (1-q_c)^2.   
\end{align*}
From this follows the left inequality of Eq. \eqref{eq:Mn2_bound}.
Finally, Eq. \eqref{eq:Mn2_limit} follows from the asymptotic distributions of $M_n^+, M_n^-$ and $M_n$ given in Theorem \ref{thm:Mnplus_asymptotic_distribution}.

$\hfill\Box$

\section{One-sided p-value computation}

Let $x_1,\ldots,x_n$ be $n$ observations with a one-sided value $M_n^+(x_1, \ldots, x_n) = c$. 
A direct approach to compute the corresponding $p$-value is to recursively evaluate the \(n-1\) integrals in  Eq. (\ref{eq:Mn_p_value})  
\begin{align}
    f_0(t) = 1, \quad f_1(t) = \int_{L_1}^t f_0(x) dx, \quad \ldots, \quad f_n(t) = \int_{L_n}^t f_{n-1}(x) dx , \label{eq:f_j}
\end{align}
where for notational simplicity we use the shorthand $L_i$ for $L_i^n(c)$.
The $p$-value of the $M_n^+$ test is then given by
\begin{equation}
    \pr{M_n^+ < c \ \Big | \mathcal{H}_0} = 1 - n! f_n(1) \,.  \label{eq:Mn_f_n}
\end{equation}
By definition, the various functions $f_d$ in Eq. (\ref{eq:f_j}) are polynomials of increasing degree, 
$f_d(x) = \sum_{k=0}^d c_{d,k} x^k$, 
whose coefficients $c_{d,0}, \ldots, c_{d,d}$ are sums of various products  of $L_1, \ldots, L_d$.  
The second column of Table \ref{tab:direct_vs_trans_polynomials} lists explicit symbolic expressions for the resulting \(f_{n}(1)\) for small values of $n$.
Clearly, the number of terms in the exact symbolic  representation grows rapidly with \(n\), and unfortunately we have not 
found a simple closed-form formula for its coefficients.
Nonetheless, one can iteratively evaluate the coefficients of the polynomials $\{f_d\}_{d=1}^n$ numerically,
since
\begin{align*}
    f_d(t)
    &=
    \int_{L_d}^t f_{d-1}(x) dx
    =
    \int_{L_d}^t \sum_{k=0}^{d-1} c_{d-1,k} x^k dx
    = \sum_{k=1}^d \frac{c_{d-1,k-1}}{k} t^k -\sum_{k=1}^d \frac{c_{d-1,k-1}}{k} L_{d}^k\,.
\end{align*}
Thus, at each iteration we store the numerical values of $c_{d,0},\ldots,c_{d,d}$ and update them 
according to the following formula
\begin{align}
    \label{eq:pvalue_naive_recursion}
    c_{d,0} = -\sum_{k=1}^{d} \frac{c_{d-1,k-1}}{k} L_{d}^k \quad \text{and} \quad \forall k\ge1: c_{d,k} = \frac{c_{d-1,k-1}}{k} \,.
\end{align}
While seemingly straightforward to evaluate, a na\"{\i}ve implementation using standard (80-bit) long double floating-point accuracy suffers from a fast accumulation of numerical errors and breaks down completely at $n \approx 150$. The heart of the problem is the formula for the constant term $c_{d,0}$ of $f_d$.
As seen from Eq. (\ref{eq:pvalue_naive_recursion}), 
at each iteration the term $c_{d+1,0}$ accumulates errors from all previous coefficients $\{c_{d,j}\}_{j=0}^d$.
These errors propagate to the higher order coefficients in the next iteration, and are again amplified when computing $c_{d+2,0}$, etc.

\begin{table}
    \begin{center}
    \begin{tabular}{ccc}
        $n$ & straightforward integration & translated polynomials \\\hline
        1 & $1-L_1$ & $1-L_1$ \\\hline
        2 & $\frac{1}{2} - L_1 - \frac{1}{2}L_2^2 + L_1 L_2$ & $\frac{1}{2}\left( 1-L_1\right)^2 - \frac{1}{2}\left(L_2-L_1\right)^2$
        \\\hline
        3 & $\frac{1}{6} - \frac{1}{2}L_1 -\frac{1}{2}L_2^2 + L_1 L_2 - \frac{1}{6} L_3^3$ & $\frac{1}{6}\left(
        1-L_1\right)^3 -\frac{1}{2}\left(L_2-L_1\right)^2 \left( 1-L_3\right)$ \\
        & $- \frac{1}{2}L_1 L_3^2  -\frac{1}{2}L_2^2 L_3 + L_1 L_2  L_3$ &  
        $- \frac{1}{6}\left(L_3-L_1\right)^3$ \\\hline

        4 & $\frac1{24} - \frac{1}{6}L_1 + \left(\frac{1}{2} L_1 L_2 - \frac{1}{4} L_2^2 \right)$ & $\frac{1}{24}\left(
        1-L_1\right)^4 -\frac{1}{4}\left(L_2-L_1\right)^2 \left( 1-L_3\right)^2$ \\
        & $- L_1 L_2 L_3 + \frac{1}{2} L_2^2 L_3 + \frac{1}{2} L_1 L_3^2 $ & $-
        \frac{1}{6}\left(L_3-L_1\right)^3 \left(1-L_4\right)-\frac{1}{24}\left( L_4-L_1\right)^4$ \\
        & $-\frac{1}{6} L_3^3 + L_1 L_2 L_3 L_4-\frac{1}{2} L_2^2 L_3 L_4$ & $ +\frac{1}{4}\left(L_2-L_1\right)^2
        \left( L_4-L_3\right)^2$ \\
        & $-\frac{1}{2} L_1 L_3^2 L_4+\frac{1}{6}
        L_3^3 L_4-\frac{1}{2} L_1 L_2 L_4^2 $ & \\
        & $+\frac{1}{4} L_2^2 L_4^2+\frac{1}{6} L_1 L_4^3 - \frac{1}{24} L_4^4$ & \\
    \end{tabular}
    \end{center}
    \caption[caption]{\label{tab:direct_vs_trans_polynomials}Comparison of symbolic expressions for $f_n(1)$ resulting from
direct integration\\ vs. computation using translated polynomials. $L_i$ is shorthand for $L_i^n(c)$.}

\end{table}


\paragraph{Computation using translated polynomials.}
To attenuate the accumulation of numerical errors we perform the calculations in a different basis for the space of degree $d$
polynomials. 
Instead of the standard basis, for each degree \(d\) we use a basis of translated monomials  $(x+t_{d,k})^k$, 
where the constants $t_{d,k}$ are yet to be determined, 
\begin{equation}
    f_d(x) = c_{d,0}+\sum_{k=1}^d c_{d,k} \left(x+t_{d,k}\right)^k \,.
        \label{eq:f_d_shifted}
\end{equation}
As in (\ref{eq:f_j}), $f_0(t)=1$, which is represented as $c_{0,0} = 1$.
 Using the representation (\ref{eq:f_d_shifted}), each integration step yields 
\begin{align*}
    f_{d}(t) &= \int_{L_{d}}^t f_{d-1}(x) dx = c_{d-1,0} (t -L_{d}) \\
    &+\sum_{k=2}^{d} \frac{c_{d-1,k-1}}{k}\left(t+t_{d-1,k-1}\right)^k
    -\sum_{k=2}^{d} \frac{c_{d-1,k-1}}{k}\left(L_{d}+t_{d-1,k-1}\right)^k \,. 
\end{align*}
Given the above form we define $t_{d,k}$ as follows,
\[
    t_{d,1} = -L_{d}, \quad \mbox{and} \quad t_{d,k} = t_{d-1,k-1}\quad  \forall k \in \{2,\ldots,d \}. 
\]
Then, the coefficients of the translated polynomial $f_d$ satisfy
\begin{align}
    \label{eq:pvalue_trans_recursion}
    c_{d,0} &= -\sum_{k=2}^{d} \frac{c_{d-1,k-1}}{k}\left(L_{d}+t_{d-1,k-1}\right)^k
    \quad\mbox{and}\quad
c_{d,k} = \frac{c_{d-1,k-1}}{k} \quad\mbox{for } k=1,\ldots,d.
\end{align}
In contrast to \eqref{eq:pvalue_naive_recursion}, in this update rule the constant term $c_{d,0}$ 
does not depend on the term $c_{d-1,0}$ of the previous iteration.
Thus, the error accumulation in this recursion is slower than in \eqref{eq:pvalue_naive_recursion}, and empirically, 
the update rule \eqref{eq:pvalue_trans_recursion} is significantly more stable.
In summary, numerical integration of \eqref{eq:pvalue_trans_recursion},
using extended double-precision (80-bit),
allows accurate calculation of one-sided $p$-values for up to $n \approx 50,000$ samples.
C++ source code for this procedure is freely available at 
\url{http://www.wisdom.weizmann.ac.il/~amitmo} 

\section*{Acknowledgements}
The authors thank Yoav Benjamini,  Ya'acov Ritov, Art Owen, Jiashun Jin, Guenther Walther
and Jonathan Rosenblatt for interesting discussions. This work was supported by a Texas A\&M-Weizmann research grant from Paul and Tina Gardner, and by a grant from the Citi Foundation.  

\bibliographystyle{chicago}
\bibliography{exact_berkjones_statistic}

\end{document}